\documentclass{article}

\usepackage{PRIMEarxiv}

\usepackage[utf8]{inputenc} 
\usepackage[T1]{fontenc}    
\usepackage{url}            
\usepackage{booktabs}       
\usepackage{amsfonts}       
\usepackage{nicefrac}       
\usepackage{microtype}      
\usepackage{fancyhdr}       
\usepackage{graphicx}       
\graphicspath{{media/}}     

\usepackage{amsmath}
\usepackage{amssymb}
\usepackage{amsthm}
\usepackage{stmaryrd}

\usepackage{xcolor}
\usepackage{tikz}
\usetikzlibrary{shapes.geometric, arrows, positioning, calc, arrows.meta}

\usepackage{hyperref}

\theoremstyle{plain}
\newtheorem{proposition}{Proposition}[section]

\theoremstyle{definition}

\newtheorem{assumption}{Assumption}[section]
\theoremstyle{remark}
\newtheorem{remark}{Remark}[section]

\DeclareMathOperator*{\argmax}{arg\,max}
\DeclareMathOperator{\Var}{Var}
\DeclareMathOperator{\Embed}{Embed}
\DeclareMathOperator{\Concat}{Concat}
\newcommand{\ctext}{c_{\text{text}}}   

\definecolor{navyblue}{RGB}{25, 55, 105}
\definecolor{burgundy}{RGB}{125, 25, 45}
\definecolor{forestgreen}{RGB}{20, 110, 60}

\title{Vision-Predictor-Action: Towards Efficient Embodied AI via Neuro-Symbolic JEPA Predictors
}

\author{
  Yeonseok Lee \\
  SLING AI Inc. \\
  Incheon, Republic of Korea \\
  \texttt{ylee@sling.ai.kr}
}

\begin{document}

\maketitle

\begin{abstract}
Vision-Language-Action (VLA) models commonly generate robot actions token by token through Large Language Model (LLM) decoders, so their inference latency grows with the length of the action sequence and limits the achievable closed-loop control frequency. We present the Goal-Conditioned Vision-Predictor-Action (VPA) framework, which replaces the autoregressive decoder with a lightweight latent control loop. A shared vision encoder and a frozen text encoder embed observations, goal images, and instructions; a neuro-symbolic selector chooses a discrete skill; a Joint-Embedding Predictive Architecture (JEPA) predictor forecasts the next latent state; and a flow-matching solver generates an entire action chunk in a few integration steps, so that the number of sequential network evaluations per decision step does not depend on the action horizon. We analyze the inference complexity of this design, its regularization against representation collapse, and the assumptions under which latent sub-goal thresholding can limit error accumulation across task phases. As an initial proof of concept, we implemented VPA in PyTorch and evaluated it on the LIBERO-10 long-horizon manipulation suite. The trainable model has fewer than $20$M parameters, and each released checkpoint is under $75$~MB, excluding a frozen $63.4$M-parameter CLIP text encoder; a decision step of a single-camera reference configuration takes approximately $10\text{--}14$~ms on an Apple M4; and, trained only on the LIBERO-10 demonstrations in under three hours on a single GPU, a policy observing the scene through two cameras succeeded in $40.0\%$ of $200$ evaluation episodes. These results come from a single training run and remain below previously reported success rates on this suite, but they suggest that the approach is workable. Code and model weights are available at \url{https://github.com/ylee-sling/Vision-Predictor-Action} and \url{https://huggingface.co/ylee-sling/vpa-libero10}, and the preprint is archived at \url{https://doi.org/10.5281/zenodo.22960920}.
\end{abstract}

\keywords{Embodied AI \and Neuro-Symbolic AI \and Joint-Embedding Predictive Architecture \and Flow Matching \and Robotic Manipulation}

\section{Introduction}

\begin{figure}[t]
\centering
\resizebox{\linewidth}{!}{%
\begin{tikzpicture}[
    font=\sffamily\small,
    rowlabel/.style={font=\sffamily\bfseries, anchor=west},
    io/.style={draw=gray!70, fill=gray!10, thick, rounded corners=3pt, minimum height=1.05cm, minimum width=2.3cm, align=center},
    llm/.style={draw=gray!70, fill=gray!25, thick, rounded corners=3pt, minimum height=1.05cm, minimum width=3.4cm, align=center},
    tok/.style={draw=gray!70, fill=white, thick, minimum size=0.6cm, inner sep=0pt},
    vpa/.style={draw=navyblue, fill=navyblue!8, thick, rounded corners=3pt, minimum height=1.05cm, minimum width=2.3cm, align=center},
    sol/.style={draw=burgundy, fill=burgundy!8, thick, rounded corners=3pt, minimum height=1.05cm, minimum width=2.3cm, align=center},
    cell/.style={draw=forestgreen, fill=forestgreen!18, minimum width=0.28cm, minimum height=1.05cm, inner sep=0pt},
    stat/.style={draw=forestgreen, fill=forestgreen!6, thick, rounded corners=4pt, minimum height=1.75cm, text width=3.65cm, align=center},
    garrow/.style={-Stealth, thick, draw=gray!60!black},
    narrow/.style={-Stealth, thick, draw=navyblue},
    note/.style={font=\sffamily\footnotesize, align=left, anchor=west}
]

\node[rowlabel, text=gray!60!black] at (0, 1.35) {Autoregressive VLA};
\node[io] (in1) at (1.15, 0) {Image +\\instruction};
\node[llm] (dec) at (4.75, 0) {Multi-billion-parameter\\LLM decoder};
\node[tok] (t1) at (7.6, 0) {$a_1$};
\node[tok] (t2) at (8.35, 0) {$a_2$};
\node[tok] (t3) at (9.1, 0) {$a_3$};
\node at (9.8, 0) {$\cdots$};
\node[tok] (tN) at (10.5, 0) {$a_N$};
\draw[garrow] (in1) -- (dec);
\draw[garrow] (dec) -- (t1);
\draw[garrow] (t3.north) to[out=135, in=45] node[above, font=\sffamily\scriptsize] {each token is fed back} ([xshift=0.9cm]dec.north);
\node[note, text=gray!60!black] at (11.25, 0) {$N = q \cdot H \cdot d_a$ sequential passes\\about 1--5\,Hz reported for RT-2~\cite{brohan2023rt}};

\node[rowlabel, text=navyblue] at (0, -1.55) {VPA (ours)};
\node[io] (in2) at (1.15, -2.85) {Images +\\instruction};
\node[vpa] (enc) at (3.9, -2.85) {Encoders\\{\scriptsize $E_\psi$ + frozen CLIP}};
\node[vpa] (sel) at (6.5, -2.85) {Skill selector\\{\scriptsize $\pi_\phi^h$}};
\node[vpa] (pre) at (9.1, -2.85) {JEPA predictor\\{\scriptsize $P_\omega$}};
\node[sol] (sol) at (11.7, -2.85) {Flow solver\\{\scriptsize $v_\theta$, $K$ steps}};
\foreach \i in {0,...,7} {\node[cell] (c\i) at (13.55 + 0.28*\i, -2.85) {};}
\node[font=\sffamily\scriptsize, align=center] at (14.53, -1.95) {whole chunk $A_t$};
\draw[narrow] (in2) -- (enc);
\draw[narrow] (enc) -- (sel);
\draw[narrow] (sel) -- (pre);
\draw[narrow] (pre) -- (sol);
\draw[-Stealth, thick, draw=burgundy] (sol) -- (c0);
\draw[navyblue, thick] ([yshift=-0.2cm]enc.south west) -- ++(0, -0.15) -- ([yshift=-0.35cm]sol.south east) -- ++(0, 0.15);
\node[font=\sffamily\footnotesize, text=navyblue] at (7.8, -4.15) {$K + 3$ sequential network evaluations per decision step, independent of the action horizon $H$};

\node[rowlabel, text=forestgreen] at (0, -5.0) {Initial results (preliminary)};
\node[stat] at (1.95, -6.25) {{\large\bfseries 10--14\,ms}\\[3pt]{\footnotesize per decision step\\on an Apple M4}};
\node[stat] at (6.0, -6.25) {{\large\bfseries \textless\,20M parameters}\\[3pt]{\footnotesize trainable; checkpoint\\of under 75\,MB}};
\node[stat] at (10.05, -6.25) {{\large\bfseries 40.0\% success}\\[3pt]{\footnotesize on LIBERO-10\\(2 cameras, 200 episodes)}};
\node[stat] at (14.1, -6.25) {{\large\bfseries \mbox{LIBERO-10 only}}\\[3pt]{\footnotesize no vision pre-training;\\\textless\,3\,h on one L4 GPU}};

\end{tikzpicture}
}
\caption{VPA at a glance. \emph{Top:} autoregressive VLAs decode actions token by token through a large language model, so the number of sequential forward passes grows with the action horizon; inference rates of about $1\text{--}5$~Hz have been reported for RT-2~\cite{brohan2023rt}. \emph{Middle:} VPA replaces the decoder with a short chain of small networks that produces an entire action chunk in $K + 3$ sequential evaluations. \emph{Bottom:} initial results of our reference implementation (Section~\ref{sec:experiments}). Latency was measured for the single-camera reference configuration (Section~\ref{subsec:latency_bound}); the parameter count and checkpoint size exclude the frozen CLIP text encoder ($63.4$M parameters), which is downloaded separately; the success rate comes from a single training run; and the models were trained only on the LIBERO-10 demonstrations. The VLA and VPA figures do not come from a controlled comparison.}
\label{fig:teaser}
\end{figure}
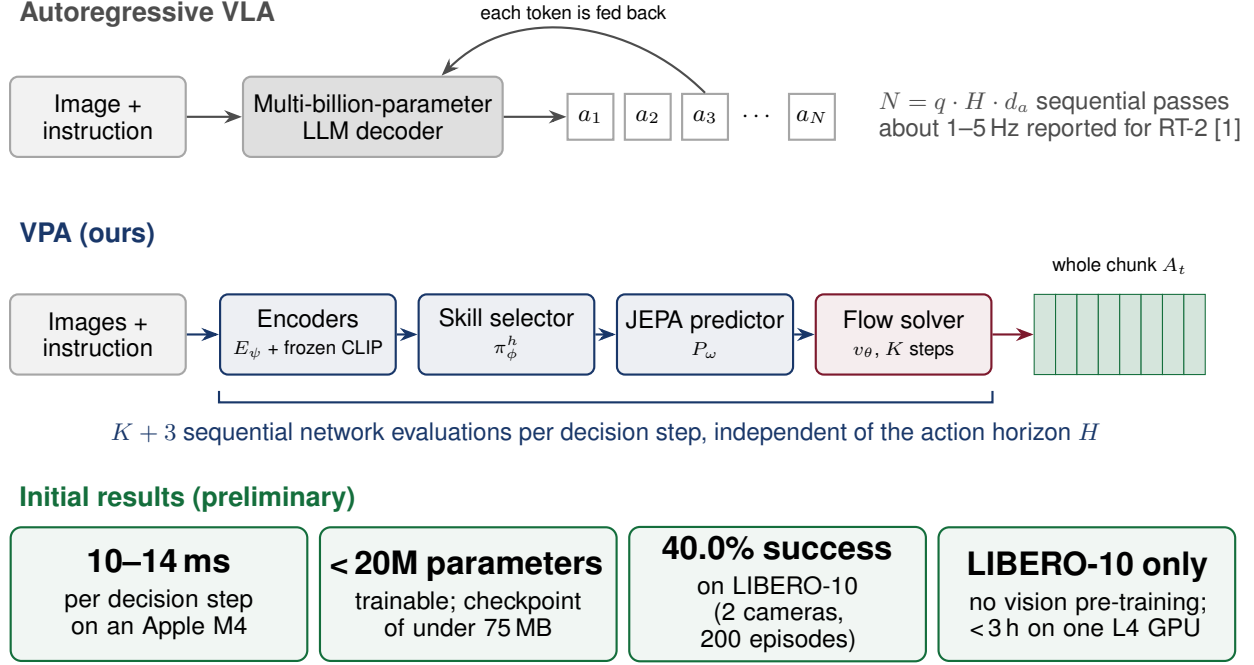

\begin{figure}[t]
\centering
\resizebox{\linewidth}{!}{%
\begin{tikzpicture}[
    node distance=0.8cm and 3.0cm,
    font=\sffamily\small,
    input/.style={draw=gray!70, fill=gray!10, thick, rounded corners=3pt, minimum width=3.4cm, minimum height=0.8cm, align=center},
    module/.style={draw=navyblue, fill=navyblue!8, thick, rounded corners=4pt, minimum width=3.6cm, minimum height=0.9cm, align=center},
    solver/.style={draw=burgundy, fill=burgundy!8, thick, rounded corners=4pt, minimum width=3.6cm, minimum height=1.1cm, align=center},
    output/.style={draw=forestgreen, fill=forestgreen!12, thick, rounded corners=4pt, minimum width=3.6cm, minimum height=0.9cm, align=center},
    arrow/.style={-Stealth, thick, draw=navyblue},
    sarrow/.style={-Stealth, thick, draw=burgundy}
]

\node[input] (visual) {Visual $I_t, I_g$ \& Text $x$};
\node[module, below=0.8cm of visual] (encoder) {\textbf{1. Multimodal Encoders}\\ViT \& CLIP Backbones};
\node[module, below=0.8cm of encoder] (selector) {\textbf{2. Primitive Selector $\pi_\phi^h$}\\Neuro-Symbolic Classifier};
\node[module, below=0.8cm of selector] (predictor) {\textbf{3. JEPA Predictor $P_\omega$}\\Latent Predictor};

\node[solver, right=3.0cm of selector] (solver) {\textbf{4. Flow-Matching Solver $\pi_\theta^l$}\\Action Trajectory Generator};
\node[input, at={(solver.north |- visual.north)}] (proprio) {Proprioception $S_t$};
\node[output, below=1.2cm of solver] (action) {\textbf{Action Chunk $A_t$}\\Executed Prefix of Length $H_t$};

\draw[arrow] (visual) -- (encoder);
\draw[arrow] (encoder) -- node[right, font=\scriptsize] {Latents $z_t, z_g^{(m_t)}, \ctext$} (selector);
\draw[arrow] (selector) -- node[right, font=\scriptsize] {Primitive $u_t$} (predictor);

\draw[sarrow] (proprio) -- node[right, font=\scriptsize] {$S_t$} (solver);
\draw[sarrow] (solver) -- (action);

\draw[sarrow] (encoder.east) -- node[above, sloped, font=\scriptsize, pos=0.45] {$z_t, \ctext$} ([yshift=0.25cm]solver.west);

\draw[sarrow] (selector.east) -- node[above, font=\scriptsize] {Token $u_t$} (solver.west);

\draw[sarrow] (predictor.east) -- node[below, sloped, font=\scriptsize, pos=0.45] {$\hat{z}_{t+1}, \sigma_{t+1}$} ([yshift=-0.25cm]solver.west);

\end{tikzpicture}
}
\caption{High-level architectural overview of the Goal-Conditioned Vision-Predictor-Action (VPA) framework. Multimodal observations are projected into latent and text embeddings that condition the selection of an active primitive $u_t$, the forecast of the next latent state $\hat{z}_{t+1}$ and its uncertainty $\sigma_{t+1}$ by a JEPA predictor, and a continuous flow-matching solver that generates an action chunk $A_t$, of which an uncertainty-dependent prefix of length $H_t$ is executed.}
\label{fig:architecture}
\end{figure}
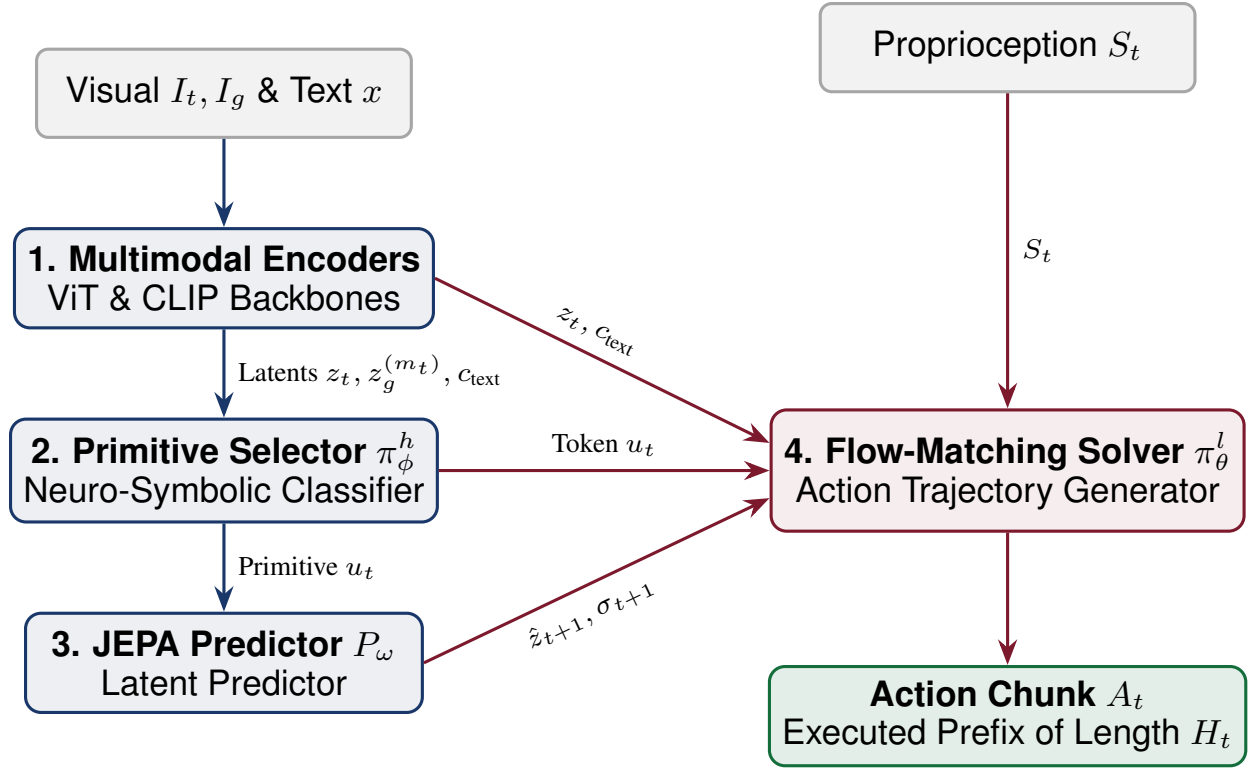

Recent advances in embodied artificial intelligence have been largely driven by Vision-Language-Action (VLA) models, which extend web-scale pre-trained backbones to physical robotic manipulation \cite{brohan2023rt, kim2024openvla}. By mapping visual observations and natural language instructions into discrete token spaces, these systems exhibit strong semantic generalization across diverse tasks \cite{huang2023instruct2act}. However, many current VLA architectures emit actions through autoregressive Large Language Model (LLM) decoders \cite{zhu2026ns}. When continuous actions are discretized into tokens, generating an $H$-step action chunk of dimension $d_a$ requires $N = q \cdot H \cdot d_a$ tokens, where $q$ denotes the number of tokens per scalar action component, and hence $\Theta(N)$ sequential forward passes through a multi-billion-parameter Transformer per execution cycle (Section~\ref{subsec:latency_bound}). For RT-2, for example, real-time inference rates of about $1\text{--}5$~Hz have been reported \cite{brohan2023rt}, which restricts the ability of autoregressive VLAs to execute reactive, dynamic maneuvers or to compensate for unexpected physical perturbations in real time.

To mitigate these latency constraints, hybrid approaches such as NS-VLA \cite{zhu2026ns} introduce discrete neuro-symbolic skill pointers, while VLA-JEPA \cite{sun2026vla} incorporates joint-embedding predictive architectures (JEPA) \cite{lecun2022path, assran2023self} to model latent world dynamics \cite{maes2026leworldmodel, zhang2026hierarchical}. Despite these improvements, such systems remain coupled to large causal vision-language backbones (e.g., Qwen3-VL-2B) that consume substantial computational resources, and their high-level plans are expressed as predefined action symbols (e.g., \textit{pick}, \textit{place\_on}) that constitute the interface between reasoning and control. Concurrently, non-autoregressive trajectory generators based on Action Chunking with Transformers (ACT) \cite{zhao2023learning}, Diffusion Policies \cite{chi2025diffusion}, and Conditional Flow-Matching (CFM) \cite{lipman2022flow, albergo2022building, black2024pi} have demonstrated high-frequency motor execution ($30\text{--}100$~Hz). These methods, however, typically either operate as reactive visual policies without forward-predictive dynamics or explicit sub-goal tracking, or, when language-conditioned as in $\pi_0$ \cite{black2024pi}, retain a large vision-language backbone for instruction processing.

In this work, we present the \textbf{Goal-Conditioned Vision-Predictor-Action (VPA)} framework, which replaces autoregressive language decoding with a non-autoregressive latent control loop. As summarized in Figure~\ref{fig:teaser}, our reference implementation, trained only on the LIBERO-10 demonstrations, has fewer than $20$M trainable parameters (a checkpoint of under $75$~MB), runs a decision step in approximately $10\text{--}14$~ms on an Apple M4 for a single-camera reference configuration, and reaches a promising initial success rate of $40.0\%$ on LIBERO-10. By substituting the LLM decoder---the dominant source of sequential latency in token-based VLA models---with a lightweight JEPA predictor, VPA reduces both the parameter count and the sequential depth of the per-step computation. Rather than communicating through text tokens, the predictor supplies forecasted state-transition latents directly to the downstream action solver, which generates continuous and potentially multi-modal action distributions within the skill selected by a discrete primitive. VPA thus retains a symbolic vocabulary for high-level skill selection, but confines it to the choice of \emph{which} skill is executed, while the motor trajectory itself is synthesized in a continuous space. We note that this design does not remove the symbolic bottleneck entirely: the set of high-level behaviors that VPA can compose is bounded by its primitive dictionary, and zero-shot generalization to skills outside this dictionary is not addressed by the present formulation. As illustrated in Figure~\ref{fig:architecture}, VPA decouples high-level task logic from low-level continuous motor control across a four-stage pipeline:
\begin{enumerate}
    \item \textbf{Latent Multimodal Goal Encoding ($E_\psi$ and Text Encoder):} A non-generative, self-supervised vision encoder \cite{dosovitskiy2020image, oquab2023dinov2} maps sensory frames into a compact latent space $\mathcal{Z} \subset \mathbb{R}^d$, in which Euclidean distances $\Vert z_t - z_g \Vert_2$ are intended to serve as proxies for physical task progress, an idea related to the latent shaping potential of NS-VLA \cite{zhu2026ns}, which uses negative Euclidean distances to learned primitive prototypes. Concurrently, a lightweight, frozen text encoder (e.g., CLIP \cite{radford2021clip}) embeds the user's natural language instruction $x$ into a dense vector $\ctext$, injecting the overarching task intent without requiring an LLM decoder.    \item \textbf{Neuro-Symbolic Primitive Selector ($\pi_\phi^h$):} A non-autoregressive classifier selects the active primitive token $u_t \in \mathcal{U}$ from a fixed dictionary of $N_u$ atomic skills, conditioned on the current latent $z_t$, the milestone latent $z_g^{(m_t)}$, and the text embedding $\ctext$. This multimodal conditioning allows the selector to disambiguate task intent (e.g., ``move the cup'' versus ``throw it away''). Milestone progression is governed by a symbolic transition rule \cite{kress2009temporal} based on the latent threshold test $\Vert z_t - z_g^{(m_t)} \Vert_2 < \tau$, where $\tau$ is scaled by the running latent variance margin maintained by the VICReg objective rather than fixed a priori (Section~\ref{subsec:primitive_selector}).
    \item \textbf{JEPA Predictor ($P_\omega$):} A lightweight predictive dynamics model forecasts the next latent state together with a scalar uncertainty estimate, $(\hat{z}_{t+1}, \sigma_{t+1}) = P_\omega(z_t, u_t)$, and is regularized against representation collapse via Variance-Invariance-Covariance Regularization (VICReg) \cite{bardes2021vicreg}. The predictor is trained with conditioning on discrete primitive tokens and subsequently frozen, so that the downstream solver is trained against a stationary conditioning distribution rather than against latent inputs that drift as the representation evolves.
    \item \textbf{Flow-Matching Action Solver ($\pi_\theta^l$):} A conditional vector field $v_\theta$ synthesizes a continuous action chunk $A_t \in \mathbb{R}^{H \times d_a}$ jointly over all $H$ steps using $K \in \{1, 2, 3\}$ numerical ODE integration steps \cite{lipman2022flow}, of which a prefix of length $H_t \le H$, adapted to the predictor's uncertainty $\sigma_{t+1}$, is executed.
\end{enumerate}

\noindent \textbf{Implementation and Initial Validation.} We implemented VPA in PyTorch and evaluated it as an initial proof of concept on LIBERO-10, the long-horizon suite of the LIBERO manipulation benchmark \cite{liu2023libero}. The implemented equations (Eqs.~\ref{eq:encoder}--\ref{eq:dynamic_horizon}, \ref{eq:variance_loss}, and~\ref{eq:covariance_loss}) were each verified against independent double-precision reference computations, and the sequence of network evaluations in every decision step was audited at runtime to confirm the sequential depth of $K + 3$ stated in Proposition~\ref{prop:complexity}. The trainable modules $E_\psi$, $\pi_\phi^h$, $P_\omega$, and $v_\theta$ together contain fewer than $20$M parameters, so that each released checkpoint occupies under $75$~MB; the frozen CLIP text encoder contributes a further $63.4$M parameters, is downloaded separately, and lies outside the per-step dependency chain, because $\ctext$ is computed once per episode. Trained only on the LIBERO-10 demonstrations, with a randomly initialized vision encoder, $K = 2$, and $H = 16$, and in under three hours on a single NVIDIA L4 GPU, the trained policy succeeded in $40.0\%$ of $200$ evaluation episodes ($20$ per task; $95\%$ Wilson interval $33.5\text{--}46.9\%$) when the observation $I_t$ comprised a third-person and a wrist-mounted camera view, compared with $10.5\%$ ($7.0\text{--}15.5\%$) with the third-person view alone. Using two views changes only the input of $E_\psi$ and leaves all subsequent equations unchanged (Section~\ref{sec:experiments}). On an Apple M4 at batch size one, a complete decision step of the single-camera reference configuration took $10.5\text{--}14.2$~ms for $K \in \{1, 2, 3\}$ and $H \le 64$.

We wish to state the scope of these results plainly. They derive from a single training run per setting and $20$ trials per task, so the success rates carry uncertainties of roughly $\pm 7$ percentage points. They remain below the success rates previously reported on this suite for existing policies, including diffusion-based policies trained without large-scale pre-training \cite{kim2024openvla}, and they do not constitute a controlled comparison with other methods. Latency was measured for the network pipeline in isolation; closed-loop control frequency on physical hardware, robustness to dynamic perturbations, and the encoder regularity constants of Section~\ref{subsec:control_stability} were not evaluated and remain future work (Section~\ref{sec:evaluation}). The experiments also used a single goal image per task and primitive labels derived heuristically from the gripper command, so the multi-milestone sequencing mechanism was not exercised. We therefore regard the present results as initial evidence that the latent control loop can be trained on standard demonstration data and executed at low latency on commodity hardware, rather than as evidence of competitive task performance.

\noindent \textbf{Open-Source Release.} To support reproduction and further investigation, we release the PyTorch reference implementation, including the training, evaluation, and latency-benchmarking code, together with the trained LIBERO-10 checkpoints; the links are given in Section~\ref{sec:availability}.

\noindent \textbf{Summary of Contributions:}
\begin{itemize}
    \item \textbf{Horizon-Independent Sequential Depth:} We show that replacing autoregressive token generation with flow-matching action chunking reduces the number of sequential network evaluations per control cycle from $\Theta(N)$, with $N = q \cdot H \cdot d_a$, to $K + 3 = \mathcal{O}(K)$, a quantity independent of the action horizon $H$. The arithmetic cost of each solver evaluation nevertheless grows with the chunk dimensionality $H \cdot d_a$; wall-clock latency is therefore approximately horizon-independent only while this computation fits within the available hardware parallelism. Latency measurements of the reference implementation were consistent with this analysis: on an Apple M4, the encoder, selector, and predictor times remained approximately constant across $H$, and the time of one solver evaluation, $t_l(H)$, remained within $10\%$ of its value at $H = 8$ up to $H = 64$ before rising by $23\text{--}27\%$ at $H = 128$. For $H \le 64$, the summed network time of Eq.~\ref{eq:vpa_latency} was $6.2\text{--}9.5$~ms, and the complete decision step, including non-network operations and timing synchronization, took $10.5\text{--}14.2$~ms.
    \item \textbf{Anti-Collapse Representation Analysis:} We characterize the conditions under which the variance hinge penalty $v(Z)$ and the covariance decorrelation penalty $c(Z)$ exclude the constant (fully collapsed) representation from the global minimizers of the joint-embedding objective, without relying on pixel-reconstruction decoders. We further use the realized variance margin to calibrate the latent completion threshold $\tau$, so that $\tau$ remains commensurate with the scale of the latent space throughout training.
    \item \textbf{Conditional Error Containment Across Task Phases:} We make explicit the encoder regularity conditions under which discrete sub-goal thresholding $\Vert z_t - z_g^{(m_t)} \Vert_2 < \tau$ limits the physical deviation at each phase transition. These are an upper Lipschitz constant $L_E$, which can be enforced with respect to the input image via spectral normalization \cite{miyato2018spectral} and governs whether the transition is triggered as the physical state approaches the milestone, and a lower (co-Lipschitz) constant $\mu_E$ with respect to task-relevant physical state, which bounds the physical deviation at transition by $\tau / \mu_E$. Because standard Vision Transformers are not globally Lipschitz \cite{kim2021lipschitz}, and because the lower constant depends on the relation between images and physical state and cannot be enforced through architectural constraints alone, we present this error containment as an empirical expectation to be verified rather than as a formal guarantee; the LIBERO-10 experiments reported in this work do not yet estimate $L_E$ or $\mu_E$.
    \item \textbf{Open-Source Reference Implementation and Proof of Concept:} We provide a PyTorch implementation of the complete pipeline, verified equation by equation against independent reference computations, and report preliminary closed-loop results on LIBERO-10, together with ablations over the camera views and the executed horizon $H_t$ (Section~\ref{sec:experiments}). The code and trained checkpoints are publicly available (Section~\ref{sec:availability}).
\end{itemize}

\subsection{Code and Model Availability}
\label{sec:availability}

The reference implementation, the trained checkpoints, and the archived preprint are publicly available:
\begin{itemize}
    \item \textbf{GitHub:} \url{https://github.com/ylee-sling/Vision-Predictor-Action} \\
    PyTorch source code for the four stages of the pipeline, the two-stage training and closed-loop LIBERO evaluation scripts, the latency benchmark of Section~\ref{subsec:latency_bound}, per-module self-tests that compare each implemented equation with an independent reference computation, and the LIBERO-10 evaluation results of Section~\ref{sec:experiments}. Evaluation videos for three trials per task and model are attached to the repository's release page.
    \item \textbf{Hugging Face:} \url{https://huggingface.co/ylee-sling/vpa-libero10} \\
    The one-camera and two-camera LIBERO-10 checkpoints reported in Section~\ref{sec:experiments}, in PyTorch and safetensors formats, together with their model configurations.
    \item \textbf{Zenodo:} \url{https://doi.org/10.5281/zenodo.22960920} \\
    The archived record of this preprint, from which later versions are also accessible.
\end{itemize}
The code and the checkpoints are released under the GNU General Public License, version~3 (GPL-3.0-only); the checkpoints were trained on the LIBERO datasets \cite{liu2023libero}, whose terms also apply. The preprint is distributed under the Creative Commons Attribution 4.0 International license (CC BY 4.0). The models were trained on Linux with an NVIDIA L4 GPU (Section~\ref{subsec:libero_setup}). The self-tests and latency benchmarks were run on macOS~15 with an Apple M4 processor (MPS backend), Python~3.13, and PyTorch~2.14.1; the multi-camera ($V = 2$) self-tests and the single-camera regression test were also run on CPU on the Linux training instance, with the CUDA build of PyTorch~2.14.1. We would be grateful for reports of errors or of attempts to reproduce the results.

\section{Related Work}
\label{sec:related_work}

\subsection{Vision-Language-Action Models}
\label{subsec:related_vla}
Recent advancements in embodied artificial intelligence have been largely driven by Vision-Language-Action (VLA) models, such as RT-2 \cite{brohan2023rt}, OpenVLA \cite{kim2024openvla}, and Instruct2Act \cite{huang2023instruct2act}. These architectures extend web-scale pre-trained backbones to map visual observations and natural language instructions into multi-modal action spaces. However, VLAs that emit discretized action tokens through an autoregressive decoder incur an inherent sequential bottleneck: generating a multi-step action chunk requires $N$ sequential forward passes through a multi-billion-parameter Transformer decoder, where $N$ grows linearly with the action horizon, resulting in $\Theta(N)$ sequential decoding steps per execution cycle. For RT-2, real-time inference rates of about $1\text{--}3$~Hz for the 55B-parameter model and about $5$~Hz for its 5B-parameter variant have been reported \cite{brohan2023rt}; such rates reduce the effectiveness of autoregressive VLAs for real-time dynamic manipulation and for compensating physical perturbations. Generalist policies that replace token decoding with continuous generative action heads, such as Octo \cite{octo2024} and $\pi_0$ \cite{black2024pi}, partially alleviate this bottleneck; we discuss these approaches in Section~\ref{subsec:related_flow}.

\subsection{Neuro-Symbolic Robotics}
\label{subsec:related_ns}
To combine high-level semantic reasoning with continuous motor execution, neuro-symbolic robotics incorporates formal logic and discrete task structures. Classical approaches utilize Linear Temporal Logic (LTL) and automata to enforce reactive mission planning and safety constraints \cite{kress2009temporal}. More recently, hybrid frameworks like NS-VLA \cite{zhu2026ns} introduce discrete skill pointers to decompose complex manipulation tasks into atomic primitives. Nevertheless, existing neuro-symbolic VLA architectures continue to rely on causal Vision-Language Models (VLMs) or autoregressive text parsers for instruction processing and high-level logic; consequently, they inherit the sequential inference latency bottleneck, which limits high-frequency closed-loop execution. Concurrent work has also begun to examine neuro-symbolic formulations of JEPA, including approaches that couple knowledge and rule embeddings with joint-embedding prediction \cite{huang2026knowledge} and hybrid symbolic-neural latent predictors \cite{huang2026sjepa}. Our framework builds upon the principles of Neuro-Symbolic Joint-Embedding Predictive Architecture (NS-JEPA) \cite{lee2026nsjepa}, replacing autoregressive language parsers with a non-autoregressive primitive selector coupled directly to a latent predictive world model. We emphasize that the symbolic component of VPA is deliberately minimal: its transition rule (Section~\ref{subsec:primitive_selector}) realizes a monotone sequential automaton over a fixed primitive dictionary rather than a general temporal-logic specification, and it therefore inherits the expressivity limits of that dictionary.

\subsection{Predictive World Models and Joint-Embedding Architectures}
\label{subsec:related_jepa}
In contrast to generative world models that reconstruct high-dimensional visual pixels, Joint-Embedding Predictive Architectures (JEPAs) \cite{lecun2022path, assran2023self} forecast future environment dynamics directly within abstract latent spaces. Because they dispense with pixel-level reconstruction decoders, JEPAs are not required to allocate model capacity to stochastic visual nuisances such as dynamic lighting and background shifts. Recent works, including LeWorldModel \cite{maes2026leworldmodel} and hierarchical latent planning frameworks \cite{zhang2026hierarchical}, construct stable, non-collapsing latent spaces for world modeling and planning. The regularizers employed in this line of work range from variance--covariance penalties such as VICReg \cite{bardes2021vicreg} to distribution-matching objectives that encourage isotropic Gaussian embeddings, as in LeJEPA \cite{balestriero2025lejepa}. At scale, V-JEPA~2 \cite{assran2025vjepa2} shows that an action-conditioned latent predictor, post-trained on a modest amount of robot interaction data, can support zero-shot manipulation through planning in representation space, albeit at the cost of iterative optimization over action sequences at inference time.

In robotic control, joint-embedding mechanics have been applied both as predictive dynamics models and as visual feature extractors. VLA-JEPA \cite{sun2026vla} introduces a unified pretraining framework with an online, time-causal latent world model to forecast future states. Conversely, JEPA-VLA \cite{miao2026jepa} argues that standard visual backbones lack temporal policy priors and employs frozen video-predictive JEPA features as plug-and-play visual representations inside existing VLA architectures via gated cross-attention. Additionally, ACT-JEPA \cite{vujinovic2026act} employs latent prediction to regularize continuous imitation learning policies. While these approaches enhance visual representations or latent forward modeling, they either lack explicit, deterministic state-machine mechanisms for sub-goal verification or remain coupled to latency-heavy, text-conditioned causal language backbones.

\subsection{Diffusion and Flow-Matching Policies}
\label{subsec:related_flow}
To enable high-frequency continuous control, non-autoregressive trajectory generation methods have emerged as an effective alternative. Action Chunking with Transformers (ACT) \cite{zhao2023learning} predicts multi-step continuous action trajectories $A_t \in \mathbb{R}^{H \times d_a}$ in a single forward pass, which reduces the effective number of decisions per episode and has been shown to mitigate compounding errors in imitation learning. Building upon this, Diffusion Policies \cite{chi2025diffusion} and Conditional Flow Matching (CFM) \cite{lipman2022flow, albergo2022building} formulate action-chunk generation as iterative denoising or as transport along a learned ODE vector field, respectively, enabling expressive multi-modal trajectory sampling; straight or nearly straight transport paths further permit accurate integration in few steps \cite{liu2022rectified}. When combined with action chunking, these continuous solvers support real-time execution rates ($30\text{--}100$~Hz). However, they typically function either as reactive visual policies that lack long-horizon predictive dynamics, symbolic verification, or explicit sub-goal state tracking, or as action heads attached to generalist transformer backbones \cite{octo2024}, which in the case of $\pi_0$ \cite{black2024pi} is a pre-trained vision-language model with billions of parameters. VPA adopts a CFM solver for low-level generation but conditions it on forecasted latent transitions and a lightweight text embedding, rather than on the hidden states of a large language backbone.

\section{Problem Formulation}
\label{sec:problem_formulation}

We formulate visual-goal robotic manipulation as a Goal-Conditioned Markov Decision Process (G-MDP) defined by the tuple $\mathcal{M} = (\mathcal{S}, \mathcal{G}, \mathcal{A}, \mathcal{P}, \mathcal{R}, \gamma_{\text{disc}})$, where $\mathcal{S}$ represents the continuous state space, $\mathcal{G}$ denotes the multimodal goal space, $\mathcal{A}$ is the continuous action-chunk space, $\mathcal{P}: \mathcal{S} \times \mathcal{A} \rightarrow \Delta(\mathcal{S})$ represents the stochastic chunk-level transition dynamics, $\mathcal{R}$ is the goal-conditioned (milestone-dependent) reward function, and $\gamma_{\text{disc}} \in [0, 1)$ is the discount factor. We write $\gamma_{\text{disc}}$ to distinguish the discount factor from the VICReg variance margin $\gamma$ introduced in Section~\ref{subsec:anti_collapse}.

\subsection{State, Observation, and Goal Spaces}
\label{subsec:state_goal}
At control step $t$, the system state $s_t \in \mathcal{S}$ pairs raw visual sensory inputs with low-level kinematic feedback:
\begin{equation}
s_t = \left( I_t, S_t \right) \in \mathcal{S}
\label{eq:state}
\end{equation}
where $I_t \in \mathbb{R}^{H_{\text{img}} \times W_{\text{img}} \times C}$ denotes the raw RGB observation frame with height $H_{\text{img}}$, width $W_{\text{img}}$, and $C$ channels, while $S_t \in \mathbb{R}^{d_s}$ represents the instantaneous robot proprioceptive state (e.g., joint positions, joint velocities, and end-effector pose).

The pair $s_t$ is an observation of an underlying physical configuration that is not directly accessible. We denote by $y_t \in \mathcal{Y} \subset \mathbb{R}^{d_y}$ the task-relevant physical state (e.g., the robot joint configuration and the poses of manipulated objects) and model image formation as $I_t = \mathcal{I}(y_t, \eta_t)$, where $\mathcal{I}$ is the image-formation (rendering) map and $\eta_t$ collects task-irrelevant nuisance factors such as illumination, background content, and sensor noise. Following common practice, we treat $s_t$ as the agent's state, implicitly assuming that it is approximately Markovian for the control task; the distinction between $y_t$ and $I_t$ is made explicit only in the analysis of Section~\ref{subsec:control_stability}.

To form a Goal-Conditioned VPA framework, the overarching task intent is specified by natural language alongside visual targets. The multimodal goal space $\mathcal{G}$ is defined by a tuple pairing a user language instruction $x$ with target goal frames $I_g$:
\begin{equation}
g = (x, I_g) \in \mathcal{G}
\label{eq:goal}
\end{equation}
The language string $x$ provides the semantic intent of the task. For multi-stage compositional tasks, the visual component $I_g$ is an ordered sequence of $M$ visual milestone goals $(I_g^{(1)}, I_g^{(2)}, \dots, I_g^{(M)})$, with $I_g^{(m)} \in \mathbb{R}^{H_{\text{img}} \times W_{\text{img}} \times C}$; single-stage tasks correspond to $M = 1$. Each milestone frame is an observation of a target physical configuration, $I_g^{(m)} = \mathcal{I}(y_g^{(m)}, \eta_g^{(m)})$, and its latent embedding is denoted $z_g^{(m)} = E_\psi(I_g^{(m)})$.

\subsection{Action Space and Action Chunking}
\label{subsec:action_space}
To account for sensorimotor latency and to reduce single-step compounding execution drift, the controller operates over continuous multi-step action chunks rather than instantaneous atomic commands. Let $\mathcal{A}_0 \subseteq \mathbb{R}^{d_a}$ denote the per-step action space, where $d_a$ is the dimensionality of the robot's action space (e.g., 7-DoF joint command vectors alongside 1-DoF parallel-gripper actuation). At step $t$, the controller selects an open-loop action chunk over horizon $H$:
\begin{equation}
A_t = (a_t, a_{t+1}, \dots, a_{t+H-1}) \in \mathcal{A} = \mathcal{A}_0^{H} \subseteq \mathbb{R}^{H \times d_a}
\label{eq:action_chunk}
\end{equation}
As detailed in Section~\ref{subsec:action_chunking}, the solver always generates a chunk of length $H$, of which an uncertainty-dependent prefix of length $H_t \le H$ is executed before the inference pipeline is invoked again. We refer to the control steps at which the pipeline is invoked as \emph{decision steps}; the primitive $u_t$, the milestone pointer $m_t$, and the conditioning vector $e_t$ are evaluated at decision steps and held fixed during chunk execution.

\subsection{Transition Dynamics and Latent Metric Reward}
\label{subsec:transition_reward}
The ground-truth physical dynamics of the environment govern state transitions at the chunk level according to:
\begin{equation}
s_{t+H} \sim \mathcal{P}\left(\cdot \mid s_t, A_t\right)
\label{eq:transition}
\end{equation}
where $\mathcal{P}$ is induced by composing the per-step dynamics over the actions of the chunk (or over its executed prefix of length $H_t$).

Rather than predicting high-dimensional pixel transitions $p(I_{t+H} \mid I_t, A_t)$ directly, the dynamics are modeled in a compact latent space $\mathcal{Z} \subset \mathbb{R}^d$ obtained via the visual encoder $E_\psi$. Within $\mathcal{Z}$, the latent transition dynamics are modeled by the forward predictor $P_\omega$ under the active discrete primitive token $u_t \in \mathcal{U}$:
\begin{equation}
z_t = E_\psi(I_t), \qquad (\hat{z}_{t+1}, \sigma_{t+1}) = P_\omega(z_t, u_t)
\label{eq:latent_dynamics}
\end{equation}
where $\sigma_{t+1} \ge 0$ is a scalar predictive uncertainty (Section~\ref{subsec:jepa_world_model}). The predictor operates at a fixed temporal stride of $\nu \in \mathbb{N}$ control steps, and we assume $\nu \le H_{\min}$, where $H_{\min}$ is the minimum executed chunk length (Section~\ref{subsec:action_chunking}), so that the forecast horizon falls within every executed chunk. For notational brevity, we write $z_{t+1}$ and $\hat{z}_{t+1}$ for the actual and forecast latents $z_{t+\nu}$ and $\hat{z}_{t+\nu}$ throughout. Because $P_\omega$ is conditioned on the primitive $u_t$ rather than on the action chunk $A_t$, it models the expected latent consequence of executing a skill, as represented in the demonstration data, rather than the consequence of a specific motor trajectory; the motor-level realization is delegated to the solver of Section~\ref{subsec:action_chunking}.

The sparse goal-reaching objective is approximated by a dense latent proxy reward based on the Euclidean distance to the active milestone:
\begin{equation}
\mathcal{R}(s_t, g; m_t) = - \left\| E_\psi(I_t) - E_\psi\big(I_g^{(m_t)}\big) \right\|_2 = - \left\| z_t - z_g^{(m_t)} \right\|_2
\label{eq:latent_reward}
\end{equation}
where $m_t \in \{1, \dots, M\}$ is the active milestone pointer. Since $\mathcal{R}$ depends on $m_t$, the pointer acts as a discrete component of an augmented state $(s_t, m_t)$. In the present formulation, $\mathcal{R}$ serves as an analysis and monitoring signal rather than as an optimization target, as all components of VPA are trained by self-supervised latent prediction or by imitation of demonstrations (Section~\ref{sec:architecture}). Whether this proxy is monotonically related to physical task progress depends on the regularity of $E_\psi$ with respect to $y_t$, which we examine in Section~\ref{subsec:control_stability}.

Milestone $m_t$ is regarded as completed when the latent distance falls below a completion threshold $\tau > 0$, such that $\| z_t - z_g^{(m_t)} \|_2 < \tau$. Because the scale of $\mathcal{Z}$ is not fixed by the joint-embedding objective, $\tau$ is calibrated relative to the running latent variance margin rather than set to an absolute constant (Section~\ref{subsec:primitive_selector}). The task is complete when this condition is satisfied for the final milestone $m_t = M$.

\section{The Goal-Conditioned VPA Architecture}
\label{sec:architecture}

\begin{figure*}[htbp]
\centering
\resizebox{0.9\textwidth}{!}{%
\begin{tikzpicture}[
    >=stealth,
    box/.style={draw, rounded corners, thick, minimum width=3.2cm, minimum height=0.9cm, align=center, fill=blue!5, font=\small},
    data/.style={draw, minimum width=2.5cm, minimum height=0.7cm, align=center, fill=gray!10, font=\small},
    concat/.style={draw, rounded corners, thick, minimum width=9cm, minimum height=0.9cm, align=center, fill=yellow!10, font=\small}
]

\node[data] (it) at (0, 0) {$I_t$ (Current Frame)};
\node[data] (ig) at (6, 0) {Milestone $I_g^{(m_t)}$ \& Text $x$};

\node[box] (enc_t) at (0, -1.3) {Encoder $E_\psi$};
\node[box] (enc_g) at (6, -1.3) {Encoder $E_\psi$ \& Text Enc};

\node[data] (zt) at (0, -2.6) {Latent $z_t$};
\node[data] (zg) at (6, -2.6) {Latent $z_g^{(m_t)}$ \& $\ctext$};

\node[box] (pi_h) at (3, -4.0) {Neuro-Symbolic Selector $\pi_\phi^h$};

\node[data] (ut) at (3, -5.3) {Token $u_t$};

\node[box] (jepa) at (0, -6.6) {Predictor $P_\omega$};
\node[data] (st) at (6, -6.6) {Proprioception $S_t$};

\node[data] (zhat) at (0, -8.1) {Predicted $(\hat{z}_{t+1}, \sigma_{t+1})$};

\node[concat] (cat) at (3, -9.6) {Unified Conditioning $e_t = \Concat\big(\Embed(u_t), z_t, \hat{z}_{t+1}, S_t, \ctext\big)$};

\node[box] (pi_l) at (3, -11.0) {Flow-Matching Solver $\pi_\theta^l$};

\node[data] (at) at (3, -12.3) {Action Chunk $A_t \in \mathbb{R}^{H \times d_a}$\\(execute first $H_t$ steps, Eq.~\ref{eq:dynamic_horizon})};

\draw[->, thick] (it) -- (enc_t);
\draw[->, thick] (ig) -- (enc_g);
\draw[->, thick] (enc_t) -- (zt);
\draw[->, thick] (enc_g) -- (zg);

\draw[->, thick, rounded corners] (zt.east) -| (pi_h.150);
\draw[->, thick, rounded corners] (zg.west) -| (pi_h.30);
\draw[->, thick] (pi_h) -- (ut);

\draw[->, thick] (zt.south) -- (jepa.north);
\draw[->, thick, rounded corners] (ut.west) -| ([xshift=1.0cm]jepa.north);
\draw[->, thick] (jepa) -- (zhat);

\draw[->, thick] (ut.south) -- (cat.north);
\draw[->, thick] (zhat.south) -- (zhat.south |- cat.north);
\draw[->, thick] (st.south) -- (st.south |- cat.north);

\draw[->, thick, dashed, rounded corners] (zt.west) -- ++(-1.5,0) |- (cat.west);
\draw[->, thick, dashed, rounded corners] (zg.east) -- ++(1.5,0) |- (cat.east);

\draw[->, thick] (cat) -- (pi_l);
\draw[->, thick] (pi_l) -- (at);

\end{tikzpicture}
} 
\caption{VPA Closed-Loop Inference Pipeline (Action Trajectory Generation).
This diagram shows the online forward pass executed at each decision step during robot deployment. The current observation $I_t$, the active milestone goal $I_g^{(m_t)}$, and the natural language instruction $x$ are encoded and used to select a primitive $u_t$, forecast $(\hat{z}_{t+1}, \sigma_{t+1})$, and form the unified conditioning vector $e_t$ (Eq.~\ref{eq:conditioning_vector}), from which the solver synthesizes an action chunk $A_t$. The predictive uncertainty $\sigma_{t+1}$ determines the executed prefix length $H_t$ (Eq.~\ref{eq:dynamic_horizon}), and the milestone pointer $m_t$ is updated by the threshold rule of Eq.~\ref{eq:state_transition}.
}
\label{fig:vpa_arch}
\end{figure*}
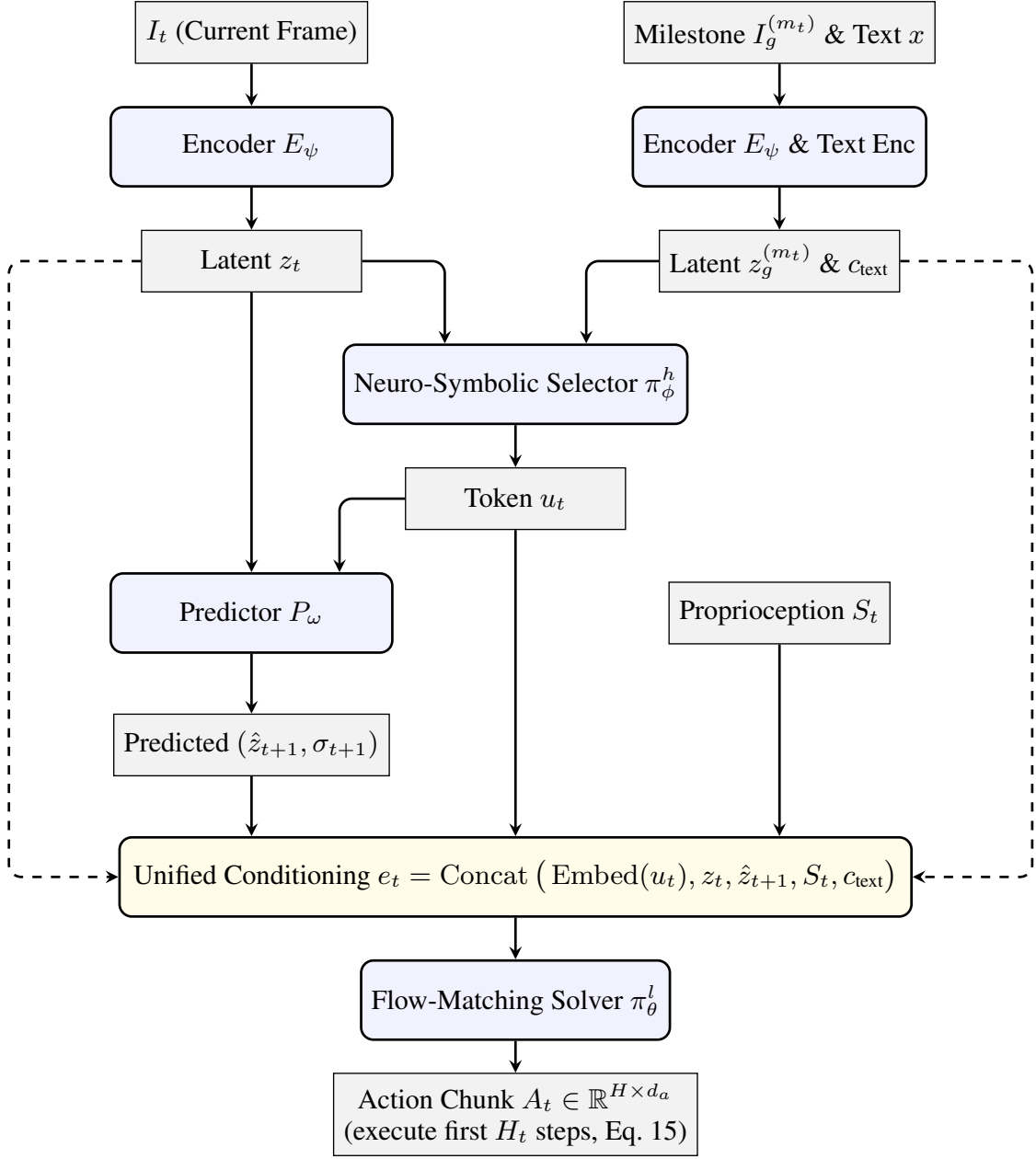

Current Vision-Language-Action (VLA) architectures that decode actions autoregressively are constrained by a sequential computing bottleneck on control frequency. To overcome this, we introduce the Goal-Conditioned Vision-Predictor-Action (VPA) framework. The architecture bypasses large language model (LLM) backbones and formulates robotic manipulation as a closed-loop latent prediction problem.

The architecture decouples high-level task logic from low-level continuous motor control through a four-stage pipeline: (1) a latent multimodal goal encoder that projects raw observations and text into a compact latent space; (2) a non-autoregressive neuro-symbolic selector that discretizes task primitives; (3) a collapse-regularized latent predictor that forecasts state transitions; and (4) a continuous flow-matching solver that outputs multi-step action trajectories. The number of sequential network evaluations per decision step in this pipeline does not depend on the action horizon (Section~\ref{subsec:latency_bound}), which keeps the per-step computation short, as reactive control requires, while retaining semantic task specification. The complete data flow is illustrated in Figure~\ref{fig:vpa_arch}.

\subsection{Latent Multimodal Goal Encoding}
\label{subsec:visual_encoding}

The perception module in the VPA framework is designed to retain task-relevant spatial and physical configurations while suppressing high-frequency visual nuisance factors $\eta_t$ (e.g., lighting variations, reflective glare, or dynamic backgrounds). Let $I_t \in \mathbb{R}^{H_{\text{img}} \times W_{\text{img}} \times C}$ denote the current observation frame at timestep $t$, and let $I_g^{(m)} \in \mathbb{R}^{H_{\text{img}} \times W_{\text{img}} \times C}$ denote the visual goal frame of milestone $m \in \{1, \dots, M\}$.

We employ a shared vision encoder $E_\psi$ (e.g., a Vision Transformer backbone \cite{dosovitskiy2020image}), which operates as a Siamese network to independently map the high-dimensional image frames into a compact, continuous latent embedding space $\mathcal{Z} \subset \mathbb{R}^d$:
\begin{equation}
    z_t = E_\psi(I_t), \qquad z_g^{(m)} = E_\psi\big(I_g^{(m)}\big), \quad m = 1, \dots, M
    \label{eq:encoder}
\end{equation}
where $z_t$ is intended to encode the current configuration of the robot end-effector and scene objects, and $z_g^{(m)}$ the target configuration of milestone $m$. The milestone embeddings depend only on the task specification and can be computed once per episode.

To capture the overarching task intent without relying on a computationally heavy LLM, we introduce a lightweight, frozen text encoder (e.g., CLIP \cite{radford2021clip} or a distilled RoBERTa \cite{liu2019roberta}). This encoder embeds the user's natural language instruction $x$ into a dense text vector $\ctext \in \mathbb{R}^{d_c}$, which is likewise computed once per episode.

Unlike standard autoencoders or generative models that minimize a pixel-level reconstruction loss $\mathcal{L}_{\text{MSE}}(I_t, \hat{I}_t)$, $E_\psi$ is optimized via a joint-embedding predictive objective (Section~\ref{subsec:jepa_world_model}). Because this objective does not require the representation to retain information that is unpredictable from one frame to the next, it encourages, but does not guarantee, invariance to the nuisance factors $\eta_t$ and a focus on predictable physical quantities such as kinematics, spatial geometry, and relative object distances. The latent distance $\Vert z_t - z_g^{(m)} \Vert_2$ is accordingly intended to serve as a proxy for physical task progress; the extent to which it does so is governed by the regularity of the composite map $y \mapsto E_\psi(\mathcal{I}(y, \eta))$ from physical state to latent, which we analyze in Section~\ref{subsec:control_stability}. In the base formulation, no explicit Lipschitz constraint is imposed on $E_\psi$; the consequences of this choice, and the role of optional spectral normalization, are discussed in the same section. The text embedding $\ctext$ provides the semantic conditioning for the discrete primitive transitions and the continuous action generation modeled in subsequent stages.

\subsection{Neuro-Symbolic Primitive Selector ($\pi_\phi^h$)}
\label{subsec:primitive_selector}

A fundamental limitation of autoregressive VLA models is their reliance on computationally heavy text generation to dictate sequential task logic. The VPA architecture instead formulates high-level reasoning as a discrete, latent-conditioned classification problem governed by a symbolic finite-state machine. The per-decision cost of this stage is a single classifier evaluation, independent of the number of milestones $M$.

Let $\mathcal{U} = \{u^{(1)}, u^{(2)}, \dots, u^{(N_u)}\}$ define a pre-determined dictionary of $N_u$ executable atomic skills, or ``primitives'' (e.g., \textit{reach}, \textit{grasp}, \textit{push}, \textit{place}). To handle multi-stage tasks, the overarching visual goal is specified as a sequence of milestones, represented in the latent space as $(z_g^{(1)}, z_g^{(2)}, \dots, z_g^{(M)})$ (Eq.~\ref{eq:encoder}).

We introduce a discrete, monotonically non-decreasing milestone pointer $m_t \in \{1, 2, \dots, M\}$ that tracks the active sub-goal at decision step $t$. The neuro-symbolic primitive selector is a lightweight classifier $\pi_\phi^h$ (parameterized by a Multi-Layer Perceptron) that maps the current state embedding, the active milestone embedding, and the text embedding $\ctext$ directly to a categorical primitive token $u_t$:
\begin{equation}
    u_t = \argmax_{u \in \mathcal{U}} \; \pi_\phi^h\big(u \mid z_t, z_g^{(m_t)}, \ctext\big)
    \label{eq:primitive_selection}
\end{equation}
The selector is trained by minimizing the cross-entropy between $\pi_\phi^h$ and the primitive labels provided by the demonstration data (Figure~\ref{fig:anti_collapse_optimization}). By incorporating $\ctext$ as a conditioning variable, the selector can disambiguate task intent (e.g., whether the instruction is to ``move the cup'' or to ``throw the cup away'') without generating text tokens. It acts as a reactive policy that selects the skill expected to satisfy both the semantic instruction and the geometric goal by reducing the latent discrepancy between $z_t$ and $z_g^{(m_t)}$.

The progression of the task is governed by a symbolic transition rule that tests sub-goal completion via an $\ell_2$-norm threshold in $\mathcal{Z}$. A threshold $\tau$ fixed in absolute units, however, is ill-suited to a jointly trained latent space: the variance hinge of the VICReg objective (Eq.~\ref{eq:variance_loss}) penalizes per-dimension standard deviations below the margin $\gamma$ but does not bound them from above, so the overall scale of $\mathcal{Z}$ can drift during training, and a static $\tau$ would become relatively looser or stricter as training proceeds. We therefore scale $\tau$ with a running estimate of the realized latent variance margin. Let $Z^{(n)} \in \mathbb{R}^{B \times d}$ denote the batch of online-encoder latents at training iteration $n$ (Section~\ref{subsec:anti_collapse}). The milestone pointer is updated at each decision step, and the threshold is calibrated during training, according to:
\begin{subequations}
\label{eq:transition_rule}
\begin{align}
    m_{t+1} &=
    \begin{cases}
        m_t + 1, & \text{if } \Vert z_t - z_g^{(m_t)} \Vert_2 < \tau \text{ and } m_t < M, \\
        m_t, & \text{otherwise},
    \end{cases}
    \label{eq:state_transition} \\
    \bar{\gamma}_n &= (1 - \alpha_\gamma)\, \bar{\gamma}_{n-1} + \alpha_\gamma \sqrt{\frac{1}{d} \sum_{j=1}^{d} \Var\big(Z^{(n)}_{\cdot, j}\big) + \epsilon},
    \label{eq:running_margin} \\
    \tau_n &= \kappa \sqrt{2d}\; \bar{\gamma}_n,
    \label{eq:dynamic_threshold}
\end{align}
\end{subequations}
where $\alpha_\gamma \in (0, 1]$ is an exponential-moving-average rate, $\epsilon$ is the numerical constant of Eq.~\ref{eq:variance_loss}, and $\kappa \in (0, 1)$ is a dimensionless tolerance. The scaling in Eq.~\ref{eq:dynamic_threshold} has a direct interpretation: for two independent latents $z_a, z_b$ drawn from the batch distribution, $\mathbb{E}\Vert z_a - z_b \Vert_2^2 = 2 \operatorname{tr}\big(\mathrm{Cov}(Z)\big)$, so that $\sqrt{2d}\,\bar{\gamma}_n$ approximates (up to $\epsilon$) the root-mean-square distance between independently drawn states, and $\kappa$ expresses the completion tolerance as a fraction of this distance. Moreover, whenever the variance hinge is inactive on the current batch, every per-dimension term satisfies $\Var(Z_{\cdot,j}) + \epsilon \ge \gamma^2$, so the instantaneous estimate in Eq.~\ref{eq:running_margin} is at least $\gamma$. Because $\bar{\gamma}_n$ is a moving average, the threshold is bounded below by $\kappa \sqrt{2d}\,\gamma$ only if the hinge was inactive on every batch that contributed to $\bar{\gamma}_n$ and $\bar{\gamma}_0 \ge \gamma$; under these conditions, the VICReg margin anchors the scale of $\tau$. Whenever the transition rule is evaluated during training (e.g., to segment demonstrations into milestone phases), $\tau = \tau_n$ is used. Once $E_\psi$ is frozen (Section~\ref{subsec:jepa_world_model}), $\bar{\gamma}$ is fixed to its final value $\bar{\gamma}^\star$, and deployment uses the constant threshold $\tau = \tau^\star = \kappa \sqrt{2d}\,\bar{\gamma}^\star$. When the test in Eq.~\ref{eq:state_transition} is satisfied for $m_t = M$, the task is declared complete and the episode terminates.

This neuro-symbolic formulation yields two structural properties. First, the selector output is restricted to the finite set $\mathcal{U}$, which precludes the generation of non-executable or out-of-vocabulary commands. Second, the threshold test turns sub-task completion into an explicit, checkable predicate on the latent state. This predicate does not by itself guarantee that the physical preconditions of the next primitive are satisfied, which depends on the regularity of $E_\psi$ (Section~\ref{subsec:control_stability}), nor does it guarantee progress: if a milestone becomes unreachable, the pointer remains at $m_t$ indefinitely. A per-milestone time-out that triggers primitive re-selection can be used to detect such stalls; its design is deferred to the evaluation phase (Section~\ref{sec:evaluation}).

\subsection{Latent Predictor ($P_\omega$)}
\label{subsec:jepa_world_model}

To generate kinematically stable and physically grounded continuous actions, the control system anticipates the consequences of its decisions.
While generative world models achieve this by predicting future pixel observations $\hat{I}_{t+1}$, pixel-level reconstruction can be computationally demanding for high-frequency control loops and may allocate model capacity toward stochastic visual nuisances (e.g., dynamic shadows or moving backgrounds) rather than relevant physical dynamics.

To address this, the VPA framework introduces a predictor $P_\omega$, parameterized as a multi-layer perceptron or causal transformer, that operates within the low-dimensional latent space $\mathcal{Z}$. Given the current latent state $z_t$ and the selected active primitive $u_t$, the predictor forecasts both the expected future latent state $\hat{z}_{t+1}$ (i.e., $\nu$ control steps ahead; Section~\ref{subsec:transition_reward}) and an associated predictive uncertainty score $\sigma_{t+1}$:
\begin{equation}
    (\hat{z}_{t+1}, \sigma_{t+1}) = P_\omega(z_t, u_t)
    \label{eq:jepa_prediction}
\end{equation}
The uncertainty score is intended to reflect epistemic uncertainty of the latent dynamics model and is used to modulate the downstream execution horizon. Because the prediction loss below does not supervise $\sigma_{t+1}$, we obtain it from the disagreement of a small ensemble \cite{lakshminarayanan2017ensembles}: $P_\omega$ comprises $N_e$ prediction heads $\{P_\omega^{(i)}\}_{i=1}^{N_e}$ with independent initializations, $\hat{z}_{t+1}$ is the ensemble mean of the head outputs $\hat{z}_{t+1}^{(i)}$, and
\begin{equation*}
    \sigma_{t+1} = \frac{1}{\sqrt{d}\,\bar{\gamma}^\star} \left( \frac{1}{N_e} \sum_{i=1}^{N_e} \big\Vert \hat{z}_{t+1}^{(i)} - \hat{z}_{t+1} \big\Vert_2^2 \right)^{1/2},
\end{equation*}
where normalization by the frozen latent scale $\bar{\gamma}^\star$ (Eq.~\ref{eq:running_margin}) renders $\sigma_{t+1}$ dimensionless, so that the sensitivity hyperparameter in Eq.~\ref{eq:dynamic_horizon} does not depend on the scale of $\mathcal{Z}$.

A known challenge in joint-embedding architectures lacking a pixel-level decoder is the vulnerability to representation collapse, in which the encoder and predictor output a constant vector to minimize prediction error. To mitigate this, $P_\omega$ and the encoder $E_\psi$ are jointly optimized using a non-collapsing prediction objective based on Variance-Invariance-Covariance Regularization (VICReg) \cite{bardes2021vicreg}. The objective consists of a latent $\ell_2$ prediction loss, which plays the role of the invariance term, coupled with a distributional regularizer:
\begin{equation}
    \mathcal{L}_{\text{JEPA}} = \mathbb{E}_{t} \left[ \big\Vert \hat{z}_{t+1} - \operatorname{sg}\big(E_{\bar{\psi}}(I_{t+1})\big) \big\Vert_2^2 \right] + \mathcal{L}_{\text{anti-collapse}}(Z), \qquad \mathcal{L}_{\text{anti-collapse}}(Z) = \lambda_v v(Z) + \lambda_c c(Z)
    \label{eq:jepa_loss}
\end{equation}
where $z_{t+1} = E_{\bar{\psi}}(I_{t+1})$ is the target embedding produced by a momentum (exponential-moving-average) copy $E_{\bar{\psi}}$ of the encoder, $\operatorname{sg}(\cdot)$ denotes the stop-gradient operator, and, with an ensemble, the prediction term is applied to each head $\hat{z}_{t+1}^{(i)}$. The regularizer is evaluated on the batch $Z$ of \emph{online}-encoder latents $z_t = E_\psi(I_t)$, since the momentum encoder receives no gradients; the variance and covariance terms $v(Z)$ and $c(Z)$ are defined in Section~\ref{subsec:anti_collapse}. These terms encourage the variance of each latent dimension to remain above the margin $\gamma$ and reduce the off-diagonal covariance to decorrelate the latent variables.

This regularized formulation encourages the prediction error $\Vert \hat{z}_{t+1} - z_{t+1} \Vert_2^2$ to capture task-relevant physical state evolutions (e.g., object displacement, gripper state changes) associated with the primitive $u_t$. Consequently, $\hat{z}_{t+1}$ provides the downstream continuous action solver with predictive physical context, substituting a forward-simulated latent state for the autoregressive reasoning of large language models.

Training proceeds in stages. In the representation-learning stage, $E_\psi$ and $P_\omega$ are trained jointly with Eq.~\ref{eq:jepa_loss}, with $P_\omega$ conditioned on the discrete primitive labels $u_t$ provided by the demonstration data. At the end of this stage, both $E_\psi$ and $P_\omega$ are frozen, which fixes the latent space, its scale $\bar{\gamma}^\star$, and the deployment threshold $\tau^\star$. Freezing provides the downstream solver with a stationary conditioning distribution rather than latent inputs that drift as the representation evolves, which simplifies the learning problem for the motor control module; it does not by itself guarantee the semantic consistency of $\mathcal{Z}$, which depends on the quality of the learned representation.

\subsection{Continuous Action Chunking via Flow-Matching ($\pi_\theta^l$)}
\label{subsec:action_chunking}

While the higher-level modules (Sections~\ref{subsec:primitive_selector} and~\ref{subsec:jepa_world_model}) operate in the semantic and latent predictive spaces, the low-level solver $\pi_\theta^l$ grounds these representations into high-frequency motor commands. To manage compounding execution errors and account for sensorimotor latency, we formulate $\pi_\theta^l$ as an action chunking policy.

At each decision step $t$, rather than predicting a single atomic action, the solver synthesizes an $H$-step continuous action trajectory chunk $A_t = (a_t, a_{t+1}, \dots, a_{t+H-1}) \in \mathbb{R}^{H \times d_a}$, where $d_a$ denotes the dimensionality of the robot's control space (e.g., 7-DoF joint velocities and 1-DoF gripper state) and $H \equiv H_{\max}$ is the fixed length of the generated chunk.

To bridge the latent dynamics with physical embodiment and task intent, we construct a unified conditioning vector $e_t$ that encapsulates intent, current state, predicted future, and real-time kinematics:
\begin{equation}
    e_t = \Concat\big(\Embed(u_t), \, z_t, \, \hat{z}_{t+1}, \, S_t, \, \ctext\big) \in \mathbb{R}^{d_e}
    \label{eq:conditioning_vector}
\end{equation}
where $\Embed: \mathcal{U} \rightarrow \mathbb{R}^{d_u}$ is a learned embedding of the discrete primitive, $z_t, \hat{z}_{t+1} \in \mathbb{R}^d$ are the current and predicted latents, $S_t \in \mathbb{R}^{d_s}$ is the instantaneous proprioceptive state of the robot, $\ctext \in \mathbb{R}^{d_c}$ is the text embedding, and $d_e = d_u + 2d + d_s + d_c$. The continuous components enter Eq.~\ref{eq:conditioning_vector} after the fixed standardization described in Section~\ref{subsec:alignment}.

Generating multimodal trajectory distributions $p(A_t \mid e_t)$ often relies on Denoising Diffusion Probabilistic Models (DDPMs) \cite{ho2020denoising, chi2025diffusion}. However, the iterative sampling procedure of DDPMs, which typically requires many denoising steps, lengthens the sequential chain of network evaluations per decision step and can therefore conflict with the control rates of $30$~Hz or more that motivate this architecture.

We therefore parameterize $\pi_\theta^l$ using a Conditional Flow-Matching (CFM) vector field $v_\theta: \mathbb{R}^{H \times d_a} \times [0, 1] \times \mathbb{R}^{d_e} \rightarrow \mathbb{R}^{H \times d_a}$, which transports a simple base distribution (a standard Gaussian) to the action-chunk distribution via an Ordinary Differential Equation (ODE) \cite{lipman2022flow}. Let $\rho \in [0, 1]$ denote the continuous integration time of the flow, distinct from the discrete control step $t$ and from the latent completion threshold $\tau$. We define a linear interpolation path between Gaussian noise $\xi \sim \mathcal{N}(0, I)$, $\xi \in \mathbb{R}^{H \times d_a}$, and the ground-truth action chunk $A_t$:
\begin{equation}
    A_t^{(\rho)} = \rho A_t + (1 - \rho)\, \xi
    \label{eq:flow_path}
\end{equation}
The neural vector field $v_\theta(A_t^{(\rho)}, \rho \mid e_t)$ is trained to regress the time-derivative of this path, i.e., the target velocity $A_t - \xi$, using the Flow-Matching objective:
\begin{equation}
    \mathcal{L}_{\text{solver}} = \mathbb{E}_{t,\, \xi,\, \rho \sim \mathcal{U}[0,1]} \left[ \big\Vert v_\theta\big(A_t^{(\rho)}, \rho \mid e_t\big) - (A_t - \xi) \big\Vert_2^2 \right]
    \label{eq:flow_loss}
\end{equation}

To mitigate the risk of physical errors caused by inaccurate latent forecasting, the system employs an uncertainty-aware execution strategy. Rather than always executing the full generated chunk, the executed horizon is scaled based on the predictor's uncertainty. We define the executed horizon $H_t$ at decision step $t$ as:
\begin{equation}
    H_t = \max \left( H_{\min}, \, \left\lfloor H_{\max} \cdot \exp\big(-\beta \bar{\sigma}_{t+1}\big) \right\rfloor \right)
    \label{eq:dynamic_horizon}
\end{equation}
where $\beta > 0$ is a sensitivity hyperparameter, $H_{\min} \le H_{\max} = H$ bound the executed chunk length, and $\bar{\sigma}_{t+1}$ is a smoothed uncertainty signal. Using the raw score $\sigma_{t+1}$ directly would allow high-frequency fluctuations in the uncertainty estimate to toggle the controller between long and short horizons on consecutive decision steps, which can induce jitter at the motor level. We therefore apply an asymmetric (fast-attack, slow-release) filter, $\bar{\sigma}_{t+1} = \max\{\sigma_{t+1}, \, (1 - \alpha_\sigma)\bar{\sigma}_t + \alpha_\sigma \sigma_{t+1}\}$ with $\alpha_\sigma \in (0, 1]$, under which the horizon contracts immediately when uncertainty rises but expands only gradually as it decays. When the environment dynamics are predictable (low $\bar{\sigma}_{t+1}$), a long prefix of the chunk is executed; during high-uncertainty events, $H_t$ contracts toward $H_{\min}$, so that the control loop relies more frequently on fresh sensory feedback $I_t$ rather than on the latent prediction.

During closed-loop execution, a chunk $\hat{A}_t \in \mathbb{R}^{H \times d_a}$ is sampled by drawing $\xi \sim \mathcal{N}(0, I)$ and integrating the learned vector field from $\rho = 0$ to $\rho = 1$ with $K$ steps of a numerical ODE solver; with the Euler method, $A^{(\rho_{k+1})} = A^{(\rho_k)} + K^{-1} v_\theta(A^{(\rho_k)}, \rho_k \mid e_t)$ for $\rho_k = k/K$, $k = 0, \dots, K-1$, and $A^{(0)} = \xi$. The first $H_t$ actions of $\hat{A}_t$ are then executed. Although the conditional paths in Eq.~\ref{eq:flow_path} are straight, the learned marginal flow can be curved, particularly for strongly multimodal action distributions; few-step integration is therefore accurate only to the extent that the learned flow is close to straight, a property that rectification procedures explicitly promote \cite{liu2022rectified}. In practice, a small number of steps $K \in \{1, 2, 3\}$ is commonly used. Each decision step thus requires $K$ sequential evaluations of $v_\theta$, each of which processes the full $H \times d_a$ chunk in parallel: the number of sequential evaluations is independent of $H$, whereas the arithmetic cost of each evaluation grows with $H \cdot d_a$ (Section~\ref{subsec:latency_bound}). In the reference implementation, a complete decision step took $10.5\text{--}14.2$~ms on an Apple M4 for $H \le 64$ (Section~\ref{subsec:latency_bound}), which is short relative to the open-loop execution of a chunk at control rates of $30\text{--}100$~Hz. Overlapping the computation of the next chunk with the execution of the current one, which such rates would ultimately require on robotic hardware, was not evaluated in this work (Section~\ref{sec:evaluation}).

\subsection{Representation Alignment and Solver Warm-Starting}
\label{subsec:alignment}

Before deploying end-to-end continuous action chunking, we adopt an initial alignment phase that matches the representation used by the frozen JEPA modules to the input expected by the continuous Flow-Matching Solver $\pi_\theta^l$.

Because the latent space is learned independently of the action data, its coordinates may have arbitrary offsets and a scale $\bar{\gamma}^\star$ that differs substantially from that of the proprioceptive and text inputs. Conditioning the vector field on unnormalized latents can produce large gradients in the early iterations of Eq.~\ref{eq:flow_loss}. After freezing, we therefore standardize the latent components of $e_t$ with fixed statistics, $\tilde{z} = (z - \bar{\mu}_Z) / \bar{\gamma}^\star$, applied to both $z_t$ and $\hat{z}_{t+1}$, where $\bar{\mu}_Z$ is the mean latent over the training data; the proprioceptive state $S_t$ is standardized analogously using dataset statistics. Because the same scale $\bar{\gamma}^\star$ also determines $\tau^\star$ and normalizes $\sigma_{t+1}$, the completion threshold, the uncertainty signal, and the solver inputs are expressed in mutually consistent units.

Utilizing principles of Behavior Cloning, the solver is then warm-started on demonstration datasets: conditioning vectors $e_t$ are computed by the frozen encoder and predictor along demonstration trajectories and paired with the corresponding action chunks, and $v_\theta$ is trained with Eq.~\ref{eq:flow_loss}. This alignment phase is intended to prevent the low-level solver from misinterpreting the latent transitions before any subsequent fine-tuning of the vector field. By establishing a shared representation early in training, the action decoder is aligned with the stationary latent space of the frozen modules, which is expected to facilitate a smoother and more reliable optimization of the flow-matching objective.

\section{Theoretical Analysis and System Properties}
\label{sec:theoretical_analysis}
\begin{figure*}[htbp]
\centering
\resizebox{0.85\textwidth}{!}{%
\begin{tikzpicture}[
    >=stealth,
    box/.style={draw, rounded corners, thick, minimum width=3.2cm, minimum height=0.9cm, align=center, fill=blue!5, font=\small},
    data/.style={draw, minimum width=2.5cm, minimum height=0.7cm, align=center, fill=gray!10, font=\small},
    concat/.style={draw, rounded corners, thick, minimum width=11cm, minimum height=0.9cm, align=center, fill=yellow!10, font=\small}
]

\node[data] (it) at (0, 0) {$I_t$ (Current Frame)};
\node[data] (it1) at (7, 0) {$I_{t+1}$ (Next Frame)};

\node[box] (enc) at (0, -1.5) {Online Encoder $E_\psi$};
\node[data] (dataset) at (3.5, -1.5) {Demonstration Data\\(Primitive Labels)};
\node[box] (enc_target) at (7, -1.5) {Target Encoder $E_{\bar{\psi}}$\\(EMA, stop-gradient)};

\node[data] (zt) at (0, -3) {Online Latent $z_t$};
\node[data] (ut) at (3.5, -3) {Primitive $u_t$};
\node[data] (zt1) at (7, -3) {Target Latent $z_{t+1}$};

\node[box] (reg) at (-3.8, -4.6) {Anti-Collapse on Online Batch $Z$\\$v(Z)$ (Eq.~\ref{eq:variance_loss}) \& $c(Z)$ (Eq.~\ref{eq:covariance_loss})};
\node[box] (predictor) at (2.5, -4.6) {JEPA Predictor $P_\omega$};

\node[data] (zhat) at (2.5, -6.2) {Predicted $\hat{z}_{t+1}$};

\node[box] (mse) at (5.5, -7.8) {Prediction Error\\$\Vert \hat{z}_{t+1} - \operatorname{sg}(z_{t+1}) \Vert_2^2$};

\node[concat] (total) at (1.5, -9.6) {Total Objective (Eq.~\ref{eq:jepa_loss}): $\mathcal{L}_{\text{JEPA}} = \mathbb{E}\Vert \hat{z}_{t+1} - \operatorname{sg}(z_{t+1}) \Vert_2^2 + \lambda_v v(Z) + \lambda_c c(Z)$};

\draw[->, thick] (it) -- (enc);
\draw[->, thick] (it1) -- (enc_target);
\draw[->, thick] (enc) -- (zt);
\draw[->, thick] (enc_target) -- (zt1);

\draw[->, thick] (dataset) -- (ut);

\draw[->, thick, rounded corners] (zt.south) |- (predictor.west);
\draw[->, thick] (ut.south) -- (ut.south |- predictor.north);
\draw[->, thick] (predictor) -- (zhat);

\draw[->, thick, rounded corners] (zhat.south) |- (mse.west);
\draw[->, thick] (zt1.south) -- (zt1.south |- mse.north);

\draw[->, thick, rounded corners] (zt.west) -| (reg.north);

\draw[->, thick] (reg.south) -- (reg.south |- total.north);
\draw[->, thick] (mse.south) -- (mse.south |- total.north);

\end{tikzpicture}
} 
\caption{Offline Representation-Learning Objective. The online encoder $E_\psi$ and the predictor $P_\omega$ are trained jointly. The target latent $z_{t+1}$ is produced by a momentum (exponential-moving-average) copy $E_{\bar{\psi}}$ of the encoder under stop-gradient, and the VICReg variance and covariance terms are evaluated on the batch $Z$ of online latents, since the target branch receives no gradients. During this stage, the predictor is conditioned on the discrete primitive labels $u_t$ supplied by logged demonstration trajectories (Eq.~\ref{eq:jepa_prediction}); the frozen text encoder is not involved. At the end of the stage, $E_\psi$ and $P_\omega$ are frozen, which fixes the latent scale $\bar{\gamma}^\star$ and the deployment threshold $\tau^\star$ (Eq.~\ref{eq:transition_rule}).}
\label{fig:anti_collapse_optimization}
\end{figure*}
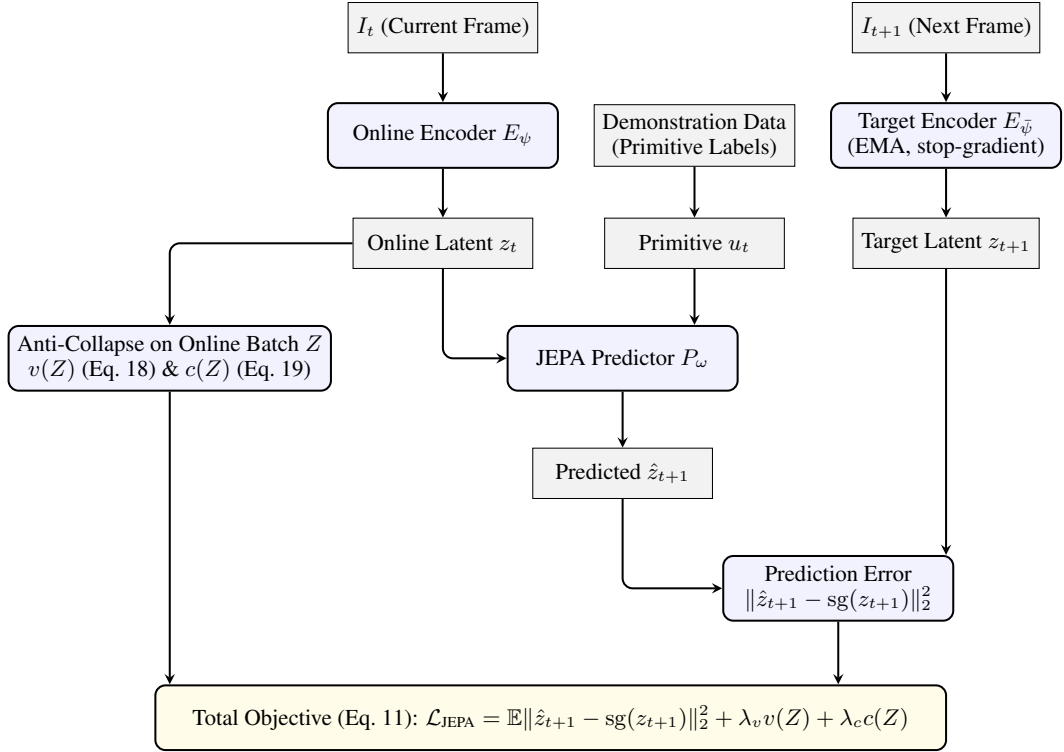

This section analyzes the runtime complexity of the VPA pipeline, the properties of its regularized latent representation, and the conditions under which its discrete sub-goal mechanism limits error accumulation across task phases. We distinguish between statements that follow directly from the definitions of the architecture and its objectives (Propositions~\ref{prop:complexity} and~\ref{prop:anti_collapse}) and statements that additionally rely on regularity properties of the learned encoder (Proposition~\ref{prop:phase_bound}). Because the latter properties cannot be enforced by the training objective or by architectural constraints alone, the corresponding conclusions are presented as conditional results and as empirical expectations to be verified in the evaluation of Section~\ref{sec:evaluation}.

\subsection{Computational Complexity and Latency}
\label{subsec:latency_bound}

A primary limitation of contemporary VLA models is their reliance on autoregressive token generation for continuous control. Let $H$ be the length of the required action chunk and $d_a$ the action dimension. Token-based VLAs discretize the continuous action space and generate $N = q \cdot H \cdot d_a$ tokens per inference cycle, where $q$ is the number of tokens per scalar action component.

For an LLM-based policy with parameters $W_{\text{LLM}}$, generating $N$ tokens requires $N$ forward passes, each of which depends on the token produced by its predecessor and therefore cannot be executed in parallel. With key-value caching, the $i$-th pass costs $\Theta(|W_{\text{LLM}}| + L_i \, d_{\text{model}})$ arithmetic operations, where $L_i$ is the attention context length at step $i$ and $d_{\text{model}}$ the hidden width. Moreover, at batch size one, each pass must read all decoder parameters from memory; for models whose parameters exceed on-chip memory, the wall-clock latency therefore satisfies
\begin{equation}
    T_{\text{VLA}} \;=\; \sum_{i=1}^{N} t_{\text{pass}}(L_i) \;\ge\; \frac{N \cdot |W_{\text{LLM}}|}{\mathrm{BW}}, \qquad N = q \cdot H \cdot d_a,
    \label{eq:vla_latency}
\end{equation}
where $t_{\text{pass}}(L_i)$ is the wall-clock time of the $i$-th pass and $\mathrm{BW}$ denotes the effective memory bandwidth, measured in parameters per unit time. The latency of autoregressive action decoding thus grows at least linearly in $N$, and hence in $H$. For RT-2, real-time inference rates of about $1\text{--}3$~Hz for the 55B-parameter model and about $5$~Hz for its 5B-parameter variant have been reported \cite{brohan2023rt}; rates in this range restrict reactive, high-speed physical adjustments.

To analyze the VPA pipeline, we distinguish two cost measures for a single decision step: the \emph{sequential depth} $D$, defined as the length of the longest chain of network evaluations in which each evaluation requires the output of its predecessor, and the total arithmetic \emph{work} $W$. Wall-clock latency is lower-bounded by the sequential depth times the minimum per-evaluation latency, and approaches the work divided by the available parallelism when the hardware is saturated.

\begin{proposition}[Sequential depth and work of a VPA decision step]
\label{prop:complexity}
Consider one decision step of VPA in which the solver generates a chunk of length $H$ with $K$ integration steps, and assume that the milestone embeddings $z_g^{(m)}$ and the text embedding $\ctext$ have been computed once per episode. Then:
\begin{enumerate}
    \item[(i)] the sequential depth is $D_{\text{VPA}} = K + 3$, which is independent of $H$ and $d_a$, whereas autoregressive decoding of $N = q \cdot H \cdot d_a$ tokens has sequential depth $D_{\text{VLA}} = N = \Theta(H \cdot d_a)$;
    \item[(ii)] the work is $W_{\text{VPA}}(H) = C_E + C_h + C_P + K \cdot C_l(H)$, where $C_E$, $C_h$, and $C_P$ are the costs of the encoder, the selector, and the predictor (including all $N_e$ ensemble heads), none of which depends on $H$, and $C_l(H)$ is the cost of one evaluation of $v_\theta$. Since $v_\theta$ outputs an element of $\mathbb{R}^{H \times d_a}$, $C_l(H) = \Omega(H \cdot d_a)$; if $v_\theta$ applies self-attention across the $H$ chunk positions, $C_l(H)$ additionally contains a term quadratic in $H$.
\end{enumerate}
\end{proposition}

\begin{proof}
(i) Within one decision step, the data dependencies form the chain $I_t \to z_t \to u_t \to (\hat{z}_{t+1}, \sigma_{t+1}) \to e_t \to A^{(\rho_1)} \to \cdots \to A^{(\rho_K)}$: the selector requires $z_t$ (Eq.~\ref{eq:primitive_selection}), the predictor requires $u_t$ (Eq.~\ref{eq:jepa_prediction}), and each Euler step requires the previous iterate (Section~\ref{subsec:action_chunking}). This chain contains one evaluation each of $E_\psi$, $\pi_\phi^h$, and $P_\omega$, and $K$ evaluations of $v_\theta$; the ensemble heads of $P_\omega$ are mutually independent and may be evaluated in parallel. No evaluation in the chain is repeated as a function of $H$. For autoregressive decoding, token $i$ is an input to pass $i+1$, so all $N$ passes lie on a single dependency chain. (ii) The expression for $W_{\text{VPA}}$ follows by summing the costs along the chain. Each evaluation of $v_\theta$ writes $H \cdot d_a$ output entries, which requires $\Omega(H \cdot d_a)$ operations.
\end{proof}

Proposition~\ref{prop:complexity} makes precise the sense in which the VPA solver removes the dependence on the action horizon. The wall-clock latency of a decision step is
\begin{equation}
    T_{\text{VPA}} = \underbrace{t_E + t_h + t_P}_{\text{independent of } H} + \; K \cdot t_l(H),
    \label{eq:vpa_latency}
\end{equation}
where $t_E$, $t_h$, $t_P$, and $t_l(H)$ denote the wall-clock times of the encoder, the selector, the predictor, and one solver evaluation, respectively. Because $K \ll N$ and the sequential dependency on output generation is removed, the number of sequential network evaluations drops from $\Theta(H \cdot d_a)$ to $K + 3 = \mathcal{O}(K)$. This does \emph{not} imply that the per-cycle cost is $\mathcal{O}(1)$ in $H$: by part (ii) of Proposition~\ref{prop:complexity}, the arithmetic work of each solver evaluation grows at least linearly with $H \cdot d_a$, and $t_l(H)$ is approximately constant in $H$ only while a single evaluation of $v_\theta$ fits within the available hardware parallelism; beyond that regime, $t_l(H)$ grows with $C_l(H)$. The reduction in sequential depth is shared with other non-autoregressive action heads (Section~\ref{subsec:related_flow}); the distinguishing latency property of VPA is that all networks on the dependency chain are lightweight relative to a vision-language backbone. The absolute values of the terms in Eq.~\ref{eq:vpa_latency} depend on the sizes of $E_\psi$, $\pi_\phi^h$, $P_\omega$, and $v_\theta$ and on the hardware, and therefore do not follow from Proposition~\ref{prop:complexity}. We measured them for the reference implementation, as described below.

\paragraph{Measured latency on an Apple M4.} We timed the four terms of Eq.~\ref{eq:vpa_latency} inside the complete decision step of the reference implementation (Section~\ref{sec:availability}), sweeping $K \in \{1, 2, 3\}$ and $H \in \{8, 16, 32, 64, 128\}$. The measurements used a Mac mini with an Apple M4 processor (10-core CPU, 32~GB memory), macOS~15.3.1, and PyTorch~2.14.1 with the MPS backend, at batch size one and in float32 arithmetic. Each stage was timed with forward hooks, with device synchronization before every clock read, and each value in Table~\ref{tab:latency} is the median of $100$ runs after $10$ warm-up runs. Because these timings do not depend on the values of the weights, randomly initialized weights were used. The benchmarked configuration uses the default, untuned module sizes of the implementation with a single camera view: $E_\psi$ is a Vision Transformer with $224 \times 224$ input, patch size $16$, width $384$, depth $6$, and latent dimension $d = 256$ ($11.12$M parameters); $\pi_\phi^h$ is a multilayer perceptron with two hidden layers of width $512$ over $N_u = 8$ primitives ($0.79$M); $P_\omega$ comprises $N_e = 5$ ensemble heads with two hidden layers of width $512$ ($2.79$M); and $v_\theta$ is a Transformer of width $256$ and depth $4$ with $d_a = 8$ ($3.55$M at $H = 8$ and $3.58$M at $H = 128$, since only its positional embedding grows with $H$). The networks on the dependency chain thus total approximately $18.3$M parameters. The frozen text encoder is not on this chain, because $\ctext$ is computed once per episode, and is excluded from the measurements.

\begin{table}[t]
\centering
\small
\caption{Measured latency of one VPA decision step on an Apple M4 (MPS backend, batch size one, float32) for the single-camera benchmark configuration described in the text. Entries are medians of $100$ runs, in milliseconds. ``Sum'' is the network time $t_E + t_h + t_P + K \cdot t_l(H)$ of Eq.~\ref{eq:vpa_latency}; ``Full step'' is the complete decision step, which additionally includes non-network operations and the synchronization overhead of the per-stage timing. Because each entry is a separate median, the columns need not add up exactly.}
\label{tab:latency}
\begin{tabular}{rrrrrrrr}
\toprule
$K$ & $H$ & $t_E$ & $t_h$ & $t_P$ & $t_l(H)$ & Sum & Full step \\
\midrule
1 & 8   & 3.732 & 0.333 & 0.692 & 1.473 & 6.229  & 10.481 \\
1 & 16  & 3.766 & 0.336 & 0.711 & 1.482 & 6.295  & 10.617 \\
1 & 32  & 3.714 & 0.334 & 0.687 & 1.485 & 6.220  & 10.484 \\
1 & 64  & 3.740 & 0.333 & 0.692 & 1.568 & 6.333  & 10.726 \\
1 & 128 & 3.867 & 0.337 & 0.715 & 1.865 & 6.783  & 11.166 \\
\addlinespace
2 & 8   & 3.756 & 0.334 & 0.701 & 1.457 & 7.707  & 12.235 \\
2 & 16  & 3.737 & 0.335 & 0.696 & 1.453 & 7.674  & 12.140 \\
2 & 32  & 3.748 & 0.329 & 0.697 & 1.509 & 7.793  & 12.266 \\
2 & 64  & 3.757 & 0.328 & 0.694 & 1.583 & 7.945  & 12.409 \\
2 & 128 & 3.752 & 0.340 & 0.707 & 1.799 & 8.397  & 12.917 \\
\addlinespace
3 & 8   & 3.703 & 0.331 & 0.694 & 1.443 & 9.057  & 13.730 \\
3 & 16  & 3.719 & 0.330 & 0.697 & 1.446 & 9.084  & 13.687 \\
3 & 32  & 3.780 & 0.337 & 0.711 & 1.524 & 9.401  & 14.205 \\
3 & 64  & 3.739 & 0.334 & 0.705 & 1.581 & 9.521  & 14.167 \\
3 & 128 & 3.772 & 0.331 & 0.702 & 1.799 & 10.201 & 14.847 \\
\bottomrule
\end{tabular}
\end{table}

\paragraph{Observations.} The measurements are consistent with Proposition~\ref{prop:complexity} in four respects. First, the implementation audits the order of network evaluations in every decision step and raises an error if it differs from one evaluation each of $E_\psi$, $\pi_\phi^h$, and $P_\omega$ followed by $K$ evaluations of $v_\theta$; this audit passed in all configurations of the sweep, so the sequential depth was $K + 3$ throughout, as stated in part~(i). Second, across all values of $K$ and $H$, $t_E$, $t_h$, and $t_P$ remained within $3.70\text{--}3.87$~ms, $0.33\text{--}0.34$~ms, and $0.69\text{--}0.72$~ms, respectively, consistent with their independence of $H$ in Eq.~\ref{eq:vpa_latency}. Third, $t_l(H)$ remained within $10\%$ of its value at $H = 8$ up to $H = 64$ and rose by $23\text{--}27\%$ at $H = 128$. This is in line with part~(ii): on this hardware, one evaluation of $v_\theta$ appears to fit within the available parallelism up to approximately $H = 64$, beyond which its growing work becomes visible in the wall-clock time. Fourth, each additional integration step lengthened the complete decision step by approximately one solver evaluation (e.g., $10.6$, $12.1$, and $13.7$~ms for $K = 1, 2, 3$ at $H = 16$). For $H \le 64$, the summed network time of Eq.~\ref{eq:vpa_latency} was $6.2\text{--}9.5$~ms, and the complete decision step took $10.5\text{--}14.2$~ms (at most $14.8$~ms at $H = 128$). At $K = 2$ and $H = 16$, the setting used for the LIBERO-10 experiments (Section~\ref{sec:experiments}), the complete decision step took $12.1$~ms. Executed back to back, these step times correspond to approximately $70\text{--}95$ decision steps per second.

\paragraph{Scope of the measurements.} We regard these results as initial empirical support for the sequential-depth analysis rather than as a characterization of deployed performance, for four reasons. (i) The benchmarked configuration uses a single $224 \times 224$ camera view, whereas the LIBERO-10 models of Section~\ref{sec:experiments} use $128 \times 128$ frames and, in the two-camera setting, two views per observation. Because both views are processed as one batch within a single evaluation of $E_\psi$, the sequential depth of the two-camera model remains $K + 3$, but its latency is not reported here. (ii) The complete decision-step times include the synchronization introduced by the per-stage timing and may therefore somewhat overstate the time of an uninstrumented step; conversely, they exclude camera capture and actuation, and thus do not measure closed-loop control frequency. (iii) Table~\ref{tab:latency} reports a single sweep on a single machine. In a repeated sweep, the median difference per timing value was $1.5\%$, $55$ of the $60$ per-stage values agreed within $5\%$, and the largest difference was $19\%$; differences of this size should be regarded as run-to-run variation. (iv) We did not run an autoregressive VLA on the same hardware. The contrast with the rates of about $1\text{--}5$~Hz reported for RT-2 \cite{brohan2023rt} is therefore indicative only, given differences in hardware, model scale, and measurement protocol, and should not be read as a measured speed-up. Measurements on robotic deployment hardware and in closed-loop operation remain future work (Section~\ref{sec:evaluation}).

\subsection{Representation Regularization and Anti-Collapse Properties}
\label{subsec:anti_collapse}

In predictive architectures trained with a mean-squared-error objective and no pixel-level decoder, the encoder and predictor admit a degenerate solution: representation collapse. The vision encoder $E_\psi$ and predictor $P_\omega$ can drive the prediction term $\mathbb{E}\big[\Vert \hat{z}_{t+1} - z_{t+1} \Vert_2^2\big]$ to zero by mapping all observations to an identical constant vector $\zeta \in \mathbb{R}^d$.

To exclude this degenerate solution, VPA incorporates the variance and covariance terms of VICReg \cite{bardes2021vicreg}. Given a batch of $B$ online-encoder latents $Z \in \mathbb{R}^{B \times d}$ evaluated during training, the anti-collapse loss acts through two penalties.

First, to counteract collapse to a point mass, a hinge loss is applied to the standard deviation of each latent dimension $j \in \{1, \dots, d\}$ with target margin $\gamma > 0$:
\begin{equation}
    v(Z) = \frac{1}{d} \sum_{j=1}^{d} \max\left(0, \gamma - \sqrt{\Var(Z_{\cdot, j}) + \epsilon} \right)
    \label{eq:variance_loss}
\end{equation}
where $\epsilon > 0$ is a small numerical constant, with $\epsilon < \gamma^2$, that prevents unbounded gradients at zero variance, and $\Var$ is computed with the same estimator as $C(Z)$ below. The term $v(Z)$ vanishes if and only if $\Var(Z_{\cdot,j}) \ge \gamma^2 - \epsilon$ for every $j$. It is a penalty rather than a hard constraint: at an optimum of the joint objective at which $v(Z) > 0$, some dimensions may remain below the margin. Furthermore, since $v(Z)$ is a batch statistic, it does not imply that $E_\psi$ is injective on the state space; physically distinct states may still be mapped to nearby or identical latents. The hinge also does not bound the per-dimension standard deviations from above, which is why the completion threshold is calibrated to the realized margin $\bar{\gamma}$ rather than to $\gamma$ alone (Eq.~\ref{eq:dynamic_threshold}).

Second, to counteract dimensional collapse, in which several latent dimensions encode the same variable, we penalize the off-diagonal entries of the empirical covariance matrix $C(Z) \in \mathbb{R}^{d \times d}$:
\begin{equation}
    c(Z) = \frac{1}{d} \sum_{i \neq j} \left( [C(Z)]_{i,j} \right)^2
    \label{eq:covariance_loss}
\end{equation}
where $C(Z) = \frac{1}{B-1} \sum_{b=1}^{B} (Z_{b, \cdot} - \bar{Z})(Z_{b, \cdot} - \bar{Z})^\top$ and $\bar{Z}$ is the batch mean. Equation~\ref{eq:covariance_loss} penalizes linear correlations between latent dimensions; it does not penalize nonlinear dependencies, so decorrelation should not be equated with statistical independence.

The regularization term is $\mathcal{L}_{\text{anti-collapse}}(Z) = \lambda_v v(Z) + \lambda_c c(Z)$ (Eq.~\ref{eq:jepa_loss}). The following proposition characterizes what this regularizer does and does not exclude. For tractability, we analyze the idealized objective in which the momentum target encoder coincides with the online encoder, $E_{\bar{\psi}} = E_\psi$, as at a fixed point of the moving-average update, and write $\mathcal{L}(\psi, \omega) = \mathcal{L}_{\text{pred}}(\psi, \omega) + \lambda_v v(Z) + \lambda_c c(Z)$, where $\mathcal{L}_{\text{pred}}$ denotes the prediction term of Eq.~\ref{eq:jepa_loss}.

\begin{proposition}[Anti-collapse properties]
\label{prop:anti_collapse}
Assume $\lambda_v > 0$ and $\epsilon < \gamma^2$.
\begin{enumerate}
    \item[(i)] Every constant encoder $E_\psi \equiv \zeta$, paired with an optimal predictor, attains the objective value $\mathcal{L}_{\text{const}} = \lambda_v \big(\gamma - \sqrt{\epsilon}\big)$. Consequently, if the model class contains parameters $(\psi', \omega')$ with $\mathcal{L}(\psi', \omega') < \lambda_v (\gamma - \sqrt{\epsilon})$, no constant encoder is a global minimizer.
    \item[(ii)] If a batch $Z$ satisfies both $v(Z) = 0$ and $c(Z) = 0$, then $C(Z)$ is diagonal with all diagonal entries at least $\gamma^2 - \epsilon > 0$; hence $C(Z)$ is positive definite and the batch latents span all $d$ dimensions of $\mathcal{Z}$. This can occur only if $B > d$.
\end{enumerate}
\end{proposition}

\begin{proof}
(i) For a constant encoder, every column of $Z$ is constant, so $\Var(Z_{\cdot, j}) = 0$ for all $j$ and $C(Z) = 0$. Each hinge term therefore equals $\max(0, \gamma - \sqrt{\epsilon}) = \gamma - \sqrt{\epsilon}$, giving $v(Z) = \gamma - \sqrt{\epsilon}$, while $c(Z) = 0$. A predictor that outputs $\zeta$ attains $\mathcal{L}_{\text{pred}} = 0$, which is optimal since $\mathcal{L}_{\text{pred}} \ge 0$. Hence $\mathcal{L}_{\text{const}} = \lambda_v(\gamma - \sqrt{\epsilon})$, and any parameters with a strictly smaller objective value exclude constant encoders from the set of global minimizers. (ii) $c(Z) = 0$ forces all off-diagonal entries of $C(Z)$ to vanish, and $v(Z) = 0$ forces $[C(Z)]_{j,j} = \Var(Z_{\cdot, j}) \ge \gamma^2 - \epsilon$ for all $j$. A diagonal matrix with positive diagonal entries is positive definite and has rank $d$; since $\operatorname{rank} C(Z) \le B - 1$, this requires $B > d$.
\end{proof}

Part (i) of Proposition~\ref{prop:anti_collapse} refines a frequently stated but imprecise claim: the positivity of $\lambda_v$ alone does not exclude collapse. Collapse is excluded from the global minimizers only if some non-degenerate representation achieves an objective value below the collapse level $\lambda_v(\gamma - \sqrt{\epsilon})$, a condition that is more readily satisfied for larger $\lambda_v$. Part (ii) shows that, at points where both penalties vanish, total and dimensional collapse are excluded on the training batch. Part (ii) does not apply to the batch size used in Section~\ref{sec:experiments}: with $B = 64$ and $d = 256$, $\operatorname{rank} C(Z) \le 63 < d$, so $v(Z)$ and $c(Z)$ cannot vanish together on a training batch. Neither part implies that latent distances faithfully reflect physical distances: a non-collapsed representation may still be insensitive to task-relevant differences or sensitive to nuisance factors. Non-collapse is therefore necessary but not sufficient for $\Vert z_t - z_g^{(m)} \Vert_2$ to serve as a proxy for task progress; the additional conditions are examined in Section~\ref{subsec:control_stability}.

\subsection{Error Accumulation and Conditional Phase-Transition Bounds}
\label{subsec:control_stability}

In step-by-step imitation learning (Behavioral Cloning), policies are subject to compounding covariate shift: if the learned policy errs with probability $\epsilon_\pi$ per step under the demonstrator's state distribution, its expected total cost over a horizon of $T$ steps (with per-step costs in $[0, 1]$) exceeds that of the demonstrator by at most $\mathcal{O}(T^2 \epsilon_\pi)$, and this quadratic dependence is tight in the worst case \cite{ross2010efficient}. VPA addresses error accumulation through two mechanisms: temporal action chunking within each task phase, and discrete sub-goal thresholding between phases.

\paragraph{Action chunking.} The Flow-Matching solver $\pi_\theta^l$ generates an $H$-step action chunk $A_t$, and executing it open-loop reduces the number of decisions over a horizon of $T$ steps from $T$ to $T/H$. Applying the same reduction at the chunk level, treating each chunk as a single decision whose cost is the sum of $H$ per-step costs and whose error probability under the demonstrator's distribution is $\epsilon_H$, yields
\begin{equation*}
    J(\hat{\pi}) - J(\pi^\ast) \;\le\; \Big(\frac{T}{H}\Big)^2 \cdot H \cdot \epsilon_H \;=\; \frac{T^2}{H}\, \epsilon_H,
\end{equation*}
where $J(\hat{\pi})$ and $J(\pi^\ast)$ are the expected total costs of the learned and demonstrator policies. This bound improves on $T^2 \epsilon_\pi$ only if $\epsilon_H$ grows sublinearly in $H$. Since a chunk contains $H$ actions, a union bound yields only $\epsilon_H \le H \epsilon_\pi$, under which the chunk-level bound coincides with the step-level bound. The benefit of chunking is therefore an empirical property, consistent with prior findings on action chunking \cite{zhao2023learning}, rather than a guarantee. Open-loop execution moreover forgoes sensory feedback within each chunk, a trade-off that the adaptive horizon of Eq.~\ref{eq:dynamic_horizon} is intended to manage.

\paragraph{Sub-goal thresholding.} The milestone pointer $m_t$ advances only when $\Vert z_t - z_g^{(m_t)} \Vert_2 < \tau$ (Eq.~\ref{eq:state_transition}). Whether this latent condition constrains the physical state depends on the regularity of the composite map from task-relevant physical state to latent, $y \mapsto E_\psi(\mathcal{I}(y, \eta))$ (Section~\ref{subsec:state_goal}). We state two conditions on this map, with a slack $\delta_\eta \ge 0$ that accounts for the latent displacement caused by variation of the nuisance factors $\eta$ within the deployment set $\mathcal{N}$; without loss of generality, the same slack is used in both.

\begin{assumption}[Upper regularity]
\label{assump:upper}
There exist $L_E < \infty$ and $\delta_\eta \ge 0$ such that, for all $y, y' \in \mathcal{Y}$ and all $\eta, \eta' \in \mathcal{N}$,
\begin{equation*}
    \big\Vert E_\psi(\mathcal{I}(y, \eta)) - E_\psi(\mathcal{I}(y', \eta')) \big\Vert_2 \;\le\; L_E \Vert y - y' \Vert_2 + \delta_\eta .
\end{equation*}
\end{assumption}

\begin{assumption}[Lower (co-Lipschitz) regularity]
\label{assump:lower}
There exist $\mu_E > 0$ and $\delta_\eta \ge 0$ such that, for all $y, y' \in \mathcal{Y}$ and all $\eta, \eta' \in \mathcal{N}$,
\begin{equation*}
    \big\Vert E_\psi(\mathcal{I}(y, \eta)) - E_\psi(\mathcal{I}(y', \eta')) \big\Vert_2 \;\ge\; \mu_E \Vert y - y' \Vert_2 - \delta_\eta .
\end{equation*}
\end{assumption}

\begin{proposition}[Conditional phase-transition bounds]
\label{prop:phase_bound}
Let $y_g^{(m)}$ denote the target physical configuration of milestone $m$, observed as $I_g^{(m)} = \mathcal{I}(y_g^{(m)}, \eta_g^{(m)})$ with $\eta_g^{(m)} \in \mathcal{N}$, and let $z_g^{(m)} = E_\psi(I_g^{(m)})$.
\begin{enumerate}
    \item[(i)] \emph{Soundness.} Under Assumption~\ref{assump:lower}, a satisfied threshold test implies physical proximity to the milestone:
    \begin{equation}
        \Vert z_t - z_g^{(m_t)} \Vert_2 < \tau \;\Longrightarrow\; \Vert y_t - y_g^{(m_t)} \Vert_2 < \frac{\tau + \delta_\eta}{\mu_E}.
        \label{eq:error_bound}
    \end{equation}
    \item[(ii)] \emph{Completeness.} Under Assumption~\ref{assump:upper} and $\tau > \delta_\eta$, sufficient physical proximity triggers the transition:
    \begin{equation*}
        \Vert y_t - y_g^{(m_t)} \Vert_2 < \frac{\tau - \delta_\eta}{L_E} \;\Longrightarrow\; \Vert z_t - z_g^{(m_t)} \Vert_2 < \tau.
    \end{equation*}
    \item[(iii)] \emph{Phase-onset reset.} Under Assumption~\ref{assump:lower}, let $t_k$ be the decision step at which the pointer advances from $k$ to $k+1$, and let $r_{k+1} = \Vert y_{t_k} - y_g^{(k)} \Vert_2$ denote the physical residual with respect to milestone $k$ at the onset of phase $k+1$. Then $r_{k+1} < (\tau + \delta_\eta)/\mu_E$, irrespective of the trajectory executed during phases $1, \dots, k$.
\end{enumerate}
\end{proposition}

\begin{proof}
(i) Applying Assumption~\ref{assump:lower} with $(y, \eta) = (y_t, \eta_t)$ and $(y', \eta') = (y_g^{(m_t)}, \eta_g^{(m_t)})$ gives $\mu_E \Vert y_t - y_g^{(m_t)} \Vert_2 - \delta_\eta \le \Vert z_t - z_g^{(m_t)} \Vert_2 < \tau$; rearranging yields Eq.~\ref{eq:error_bound}. (ii) Applying Assumption~\ref{assump:upper} to the same pair gives $\Vert z_t - z_g^{(m_t)} \Vert_2 \le L_E \Vert y_t - y_g^{(m_t)} \Vert_2 + \delta_\eta < (\tau - \delta_\eta) + \delta_\eta = \tau$. (iii) At $t_k$, the threshold test is satisfied for milestone $k$ by the definition of Eq.~\ref{eq:state_transition}; part (i) then gives the bound, which depends only on $\tau$, $\delta_\eta$, and $\mu_E$.
\end{proof}

\begin{remark}[Direction of the Lipschitz inequality]
\label{rem:direction}
An upper bound of the form $\Vert z_a - z_b \Vert_2 \le L_E \Vert y_a - y_b \Vert_2$ (Assumption~\ref{assump:upper} with $\delta_\eta = 0$) implies only $\Vert y_t - y_g^{(m_t)} \Vert_2 \ge \Vert z_t - z_g^{(m_t)} \Vert_2 / L_E$, which is a \emph{lower} bound on the physical deviation. An upper Lipschitz condition alone therefore cannot bound the physical error at a transition from above; that conclusion requires the lower regularity of Assumption~\ref{assump:lower}. The two assumptions play complementary roles: Assumption~\ref{assump:lower} makes the threshold test sound (a transition implies physical proximity), whereas Assumption~\ref{assump:upper} makes it complete (physical proximity triggers a transition).
\end{remark}

\paragraph{Consequences and their limits.} Under Assumption~\ref{assump:lower}, part (iii) of Proposition~\ref{prop:phase_bound} implies $\max_k r_k < (\tau^\star + \delta_\eta)/\mu_E$ over all phases that are entered: the physical residual with which each primitive begins is bounded uniformly and independently of the length $M$ of the milestone sequence, so that deviation accumulated in one phase is not propagated beyond this bound into the next. Three qualifications apply. First, the bound concerns phase onsets only; deviations within a phase are governed by the low-level policy and are subject to the chunk-level analysis above. Second, the bound is conditional on each transition occurring: by part (ii), a transition is guaranteed only if the physical state enters the ball of radius $(\tau^\star - \delta_\eta)/L_E$ around the milestone, which requires $\tau^\star > \delta_\eta$, that is, a threshold exceeding the latent displacement induced by nuisance variation; otherwise the pointer may stall (Section~\ref{subsec:primitive_selector}). Third, $\tau^\star$ and $\delta_\eta$ jointly determine a trade-off: increasing $\kappa$ in Eq.~\ref{eq:dynamic_threshold} raises $\tau^\star$, which improves completeness under nuisance variation but loosens the soundness bound of Eq.~\ref{eq:error_bound}.

\paragraph{Why these properties are expectations rather than guarantees.} The constants $L_E$, $\mu_E$, and $\delta_\eta$ are properties of the learned encoder composed with image formation, and none of them is controlled by the objective of Eq.~\ref{eq:jepa_loss}. Regarding Assumption~\ref{assump:upper}, standard dot-product self-attention is not globally Lipschitz \cite{kim2021lipschitz}, so a Vision Transformer encoder possesses no global Lipschitz bound with respect to its input unless explicit constraints are imposed. On a bounded image domain a finite constant exists, but it may be very large and is neither estimated nor controlled during training. Spectral normalization of the linear layers \cite{miyato2018spectral}, combined with a Lipschitz-continuous attention variant \cite{kim2021lipschitz}, can bound the pixel-space Lipschitz constant $L_\psi$ of $E_\psi$. Even then, Assumption~\ref{assump:upper} concerns the composition with the image-formation map $\mathcal{I}$, whose local sensitivity to $y$ can be very large at occlusion and contact events, so that $L_E$ is bounded by $L_\psi$ times the Lipschitz constant of $\mathcal{I}$ only where the latter is finite. Spectral normalization is not part of the base formulation; its effect on representation quality and on the threshold mechanism is a subject for empirical evaluation. Regarding Assumption~\ref{assump:lower}, a lower bound cannot be obtained from architectural constraints of this kind at all. The encoder is designed to be non-injective with respect to pixels, being intended to be invariant to $\eta$, so the lower bound must hold with respect to $y$ alone. It therefore depends on the data and on the learned representation, and it is violated wherever task-relevant differences are not visible in the image, for example under occlusion or object symmetry. The variance term of Eq.~\ref{eq:variance_loss} promotes only batch-level spread and does not imply Assumption~\ref{assump:lower}.

Accordingly, we do not claim that VPA guarantees non-compounding error. We instead state the expectation that, to the extent that the learned encoder satisfies Assumptions~\ref{assump:upper} and~\ref{assump:lower} with favorable constants on the deployment distribution, discrete sub-goal thresholding limits the propagation of physical error across task phases in the manner described by Proposition~\ref{prop:phase_bound}. This expectation is directly testable. In simulation, where the physical state $y_t$ is available, the constants $\mu_E$, $L_E$, and $\delta_\eta$ can be estimated from the lower and upper envelopes of $\Vert z - z' \Vert_2$ as a function of $\Vert y - y' \Vert_2$ over sampled state pairs and nuisance conditions, and the phase-onset residuals $r_k$ can be measured directly (Section~\ref{sec:evaluation}).

\section{Empirical Evaluation}
\label{sec:experiments}

This section reports a first empirical evaluation of the reference implementation on the LIBERO-10 benchmark \cite{liu2023libero}. The experiments are intended as a proof of concept: they test whether the pipeline of Section~\ref{sec:architecture} can be trained on standard demonstration data and executed in closed loop, and they measure the effect of two design choices, the camera configuration and the executed horizon $H_t$. Each configuration was trained once and evaluated on $200$ episodes, so that success rates carry $95\%$ intervals of roughly $\pm 7$ percentage points, and no baselines were evaluated under the same protocol. The results should therefore be read as preliminary. The trained checkpoints, the evaluation summaries, and the code used to produce them are publicly available (Section~\ref{sec:availability}).

\subsection{LIBERO-10 Implementation}
\label{subsec:libero_setup}

\paragraph{Benchmark and data.} LIBERO-10 is the long-horizon suite of LIBERO \cite{liu2023libero}: ten tabletop manipulation tasks in kitchen, living-room, and study scenes, each specified by a natural-language instruction, most of which combine two sub-tasks (e.g., placing two objects in a basket, or placing an object in a drawer and closing it). We trained on the demonstrations provided with the suite, holding out two demonstrations per task ($20$ in total) for validation. Frames were rotated by $180^\circ$, because LIBERO renders them upside down, and used at their stored resolution of $128 \times 128$ pixels. The proprioceptive state $S_t$ concatenated the end-effector position and orientation, the gripper state, and the joint states ($d_s = 15$), and the action space was LIBERO's six-dimensional end-effector displacement with a gripper command ($d_a = 7$). We verified on the dataset files that each recorded frame follows the action with the same index, so that each training pair associates the observation $(I_t, S_t)$ with the chunk of the $H$ subsequent actions.

\paragraph{Primitive labels and milestones.} The formulation of Section~\ref{subsec:primitive_selector} assumes primitive labels supplied with the demonstrations, which LIBERO does not provide. As a stand-in, we derived $N_u = 4$ primitives from the recorded gripper command with a window of $10$ steps: \emph{grasp} for the $10$ steps after the gripper closes, \emph{carry} while it remains closed thereafter, \emph{release} for the $10$ steps after it opens, and \emph{reach} otherwise. This heuristic is not part of the method and was chosen for simplicity; annotated or learned segmentations may behave differently. Each task was specified by a single milestone ($M = 1$): during training, the goal image $I_g^{(1)}$ was the final frame of the same demonstration, and at evaluation, trial $k$ of each task used the final frame of demonstration $k \bmod 50$ in that task's LIBERO demonstration file, so that the $20$ trials used demonstrations $0\text{--}19$, most of which were also training demonstrations. With $M = 1$, the milestone pointer of Eq.~\ref{eq:state_transition} never advances between phases, so the multi-phase sequencing mechanism and the phase-transition analysis of Section~\ref{subsec:control_stability} were not exercised by these experiments. We further note that the goal image provides the policy with information about the target arrangement that language-conditioned policies evaluated on LIBERO do not receive.

\paragraph{Training.} Following Section~\ref{subsec:jepa_world_model}, training proceeded in two stages of $50{,}000$ steps each with batch size $64$. In the first stage, $E_\psi$ and $P_\omega$ were trained from random initialization with Eq.~\ref{eq:jepa_loss}, after which both were frozen, which fixed $\bar{\gamma}^\star$ and $\tau^\star$. In the second stage, every frame was encoded once by the frozen encoder, the conditioning inputs were standardized as described in Section~\ref{subsec:alignment}, and the selector and the solver were trained with the cross-entropy loss of Section~\ref{subsec:primitive_selector} and with Eq.~\ref{eq:flow_loss}, respectively, using the demonstration primitive labels as inputs to $\Embed(u_t)$ and $P_\omega$. The two-camera model differs from the one-camera model only in the input of $E_\psi$: both views are passed through the same Vision Transformer body, and a single linear head maps their concatenated class tokens to $z_t$, so that Eqs.~\ref{eq:primitive_selection}--\ref{eq:dynamic_horizon} are unchanged. The trainable modules of both models contain fewer than $20$M parameters; the frozen CLIP text encoder ($63.4$M parameters) is evaluated once per episode. All settings are listed in Table~\ref{tab:libero_setup}; they are default values of the implementation and were not tuned. Both models were trained on a Google Cloud \texttt{g2-standard-4} instance with a single NVIDIA L4 GPU (24~GB), four vCPUs, and 16~GB of memory. Training took under $2$~hours of wall-clock time for the one-camera model and under $3$~hours for the two-camera model. Most of this time was spent in the first stage; the second stage, which operates on cached latents, took well under an hour.

\begin{table}[t]
\centering
\small
\caption{Settings of the LIBERO-10 experiments. The listed values are defaults of the reference implementation and were not tuned; both models share them except for the camera views.}
\label{tab:libero_setup}
\begin{tabular}{lp{0.62\linewidth}}
\toprule
Component & Setting \\
\midrule
Observation $I_t$ & third-person view, or third-person and wrist views; $128 \times 128$ RGB \\
Proprioception $S_t$ & end-effector position and orientation, gripper and joint states ($d_s = 15$) \\
Actions & six-dimensional end-effector displacement and gripper command ($d_a = 7$) \\
Encoder $E_\psi$ & Vision Transformer, patch size $16$, width $384$, depth $6$, $6$ heads, $d = 256$; random initialization \\
Text encoder & frozen CLIP ViT-B/32 text encoder, $\ctext \in \mathbb{R}^{512}$ \\
Primitives & $N_u = 4$ (reach, grasp, carry, release), derived from the gripper command \\
Milestones & $M = 1$ (final frame of a demonstration) \\
Selector $\pi_\phi^h$ & MLP with two hidden layers of width $512$ \\
Predictor $P_\omega$ & $N_e = 5$ heads, each an MLP with two hidden layers of width $512$; $\nu = 1$ \\
Solver $v_\theta$ & Transformer, width $256$, depth $4$, $4$ heads; $H = 16$, $H_{\min} = 4$, $K = 2$ \\
Eq.~\ref{eq:jepa_loss} & $\lambda_v = 25$, $\lambda_c = 1$, $\gamma = 1$, $\epsilon = 10^{-4}$, momentum of $E_{\bar{\psi}}$: $0.996$ \\
Eqs.~\ref{eq:running_margin}, \ref{eq:dynamic_threshold}, \ref{eq:dynamic_horizon} & $\alpha_\gamma = 0.01$, $\kappa = 0.1$, $\alpha_\sigma = 0.1$, $\beta = 1$ \\
Optimization & AdamW with warm-up and cosine decay; batch size $64$; $50$k $+$ $50$k steps; seed $0$ \\
\bottomrule
\end{tabular}
\end{table}

\paragraph{Evaluation protocol.} Each model was evaluated in closed loop in the LIBERO simulator on $20$ trials per task ($200$ episodes per configuration), starting from LIBERO's fixed initial states with evaluation seed $0$. After each reset, five zero-action steps were executed before control started. An episode was counted as a success if LIBERO's task-completion predicate was satisfied within $600$ control steps. Although Section~\ref{subsec:primitive_selector} specifies that an episode terminates once the test of Eq.~\ref{eq:state_transition} holds for $m_t = M$, this completion flag did not end episodes in our experiments, and success was judged only by LIBERO's predicate. Unless stated otherwise, the executed horizon followed Eq.~\ref{eq:dynamic_horizon} with $\beta = 1$, which in practice yielded $H_t = 15$ of the $H = 16$ generated actions (Section~\ref{subsec:horizon_ablation}). We report success rates with $95\%$ Wilson score intervals. Because each configuration was trained with a single seed, these intervals reflect evaluation variability only, not variability across training runs.

\subsection{Task Success Rates (One vs. Two Cameras)}
\label{subsec:camera_results}

Table~\ref{tab:camera_results} reports success rates for the one-camera model, which observes only the third-person view, as in the single-frame formulation of Section~\ref{subsec:visual_encoding}, and for the two-camera model, which additionally observes the wrist-mounted view. The one-camera model succeeded in $10.5\%$ of episodes ($21/200$; $95\%$ interval $7.0\text{--}15.5\%$) and failed in all trials of six of the ten tasks. Adding the wrist view raised the success rate to $40.0\%$ ($80/200$; $33.5\text{--}46.9\%$). The two aggregate intervals do not overlap, and the two-camera model achieved a higher rate on each of the ten tasks. With $20$ trials per task, however, individual per-task rates carry intervals of up to about $\pm 20$ percentage points, and we do not draw conclusions about particular tasks.

\begin{table}[t]
\centering
\small
\caption{LIBERO-10 success rates (\%) of the one-camera (third-person view) and two-camera (third-person and wrist views) models. Each task was evaluated on $20$ trials; the aggregate row lists $95\%$ Wilson intervals in brackets. One training run per model.}
\label{tab:camera_results}
\begin{tabular}{rlcc}
\toprule
& Task & One camera & Two cameras \\
\midrule
0 & Alphabet soup and tomato sauce into the basket & $0$ & $10$ \\
1 & Cream cheese box and butter into the basket & $0$ & $45$ \\
2 & Turn on the stove and put the moka pot on it & $30$ & $40$ \\
3 & Black bowl into the bottom drawer, then close it & $35$ & $55$ \\
4 & White mug on the left plate, yellow-and-white mug on the right plate & $0$ & $40$ \\
5 & Book into the back compartment of the caddy & $20$ & $50$ \\
6 & White mug on the plate, chocolate pudding to its right & $0$ & $15$ \\
7 & Alphabet soup and cream cheese box into the basket & $0$ & $50$ \\
8 & Both moka pots on the stove & $0$ & $45$ \\
9 & Yellow-and-white mug into the microwave, then close it & $20$ & $50$ \\
\midrule
& All tasks ($200$ episodes) & $10.5$ {\footnotesize $[7.0, 15.5]$} & $40.0$ {\footnotesize $[33.5, 46.9]$} \\
\bottomrule
\end{tabular}
\end{table}

The size of this effect suggests that, in this configuration ($128 \times 128$ frames and an encoder trained from random initialization on the demonstrations alone), the third-person view did not provide the policy with sufficient information on its own. One plausible explanation, which we have not tested directly, is that the wrist view resolves the relation between the gripper and nearby objects more clearly at this resolution. The result also indicates that the single-frame observation assumed in Section~\ref{subsec:visual_encoding} was insufficient in this setting, and that treating $I_t$ as a set of views, which changes only the input of $E_\psi$, is a useful extension of the formulation.

\paragraph{Dependence on the third-person view.} Replacing the third-person view of the two-camera model with a black image at every decision step, while leaving the goal image unchanged, reduced its success rate to $0\%$ ($0/200$; upper $95\%$ bound $1.9\%$). Because a black frame lies outside the training distribution, this test shows that the policy relies on the third-person view, but it does not cleanly measure how much task-relevant information that view contributes.

\paragraph{Training length.} Evaluating the two-camera model at intermediate checkpoints after $40$k, $45$k, and $50$k second-stage steps gave $32.5\%$ ($26.4\text{--}39.3\%$), $36.0\%$ ($29.7\text{--}42.9\%$), and $40.0\%$, respectively. The differences between adjacent checkpoints are within the statistical uncertainty, but the trend suggests that the policy had not converged when training stopped.

\paragraph{Context.} Previously reported success rates on LIBERO-10 are higher than those reported here, including for policies trained without large-scale pre-training \cite{kim2024openvla}. Our protocol differs from the standard one in the number of trials and training seeds and, as noted in Section~\ref{subsec:libero_setup}, in the goal image supplied to the policy, so we do not make a direct comparison. We regard the two-camera result as evidence that the pipeline can learn non-trivial long-horizon behavior from demonstrations, not as evidence of competitive performance.

\subsection{Execution Horizon Ablation}
\label{subsec:horizon_ablation}

The executed horizon $H_t$ determines how many of the $H = 16$ generated actions are executed open loop before the next decision step. Using the final two-camera checkpoint, we compared fixed executed horizons of $2$, $4$, and $8$ actions with Eq.~\ref{eq:dynamic_horizon} at two sensitivities: $\beta = 1$, for which the filtered uncertainty $\bar{\sigma}_{t+1}$ remained small enough that $H_t = 15$, and $\beta = 40.4$, chosen so that the 90th percentile of the per-step uncertainty $\sigma_{t+1}$ observed in the $\beta = 1$ run gives $H_t = H_{\min}$, for which $H_t$ varied with $\bar{\sigma}_{t+1}$ and was usually much shorter than the chunk. All variants used the same checkpoint and the same $200$ initial states. Table~\ref{tab:horizon_ablation} shows the results.

\begin{table}[t]
\centering
\small
\caption{Success rate of the two-camera model on LIBERO-10 as a function of the number of actions executed per decision step ($H = 16$). All rows use the same checkpoint and the same $200$ episodes; intervals are $95\%$ Wilson intervals. The value $\beta = 40.4$ was chosen so that the 90th percentile of $\sigma_{t+1}$ in the $\beta = 1$ run gives $H_t = H_{\min}$; it was not tuned for task success.}
\label{tab:horizon_ablation}
\begin{tabular}{lccc}
\toprule
Executed horizon & Successes & Success (\%) & $95\%$ interval \\
\midrule
Fixed, $H_t = 2$ & $32/200$ & $16.0$ & $11.6\text{--}21.7$ \\
Fixed, $H_t = 4$ & $52/200$ & $26.0$ & $20.4\text{--}32.5$ \\
Fixed, $H_t = 8$ & $67/200$ & $33.5$ & $27.3\text{--}40.3$ \\
Eq.~\ref{eq:dynamic_horizon}, $\beta = 40.4$ (adaptive $H_t$) & $61/200$ & $30.5$ & $24.5\text{--}37.2$ \\
Eq.~\ref{eq:dynamic_horizon}, $\beta = 1$ ($H_t = 15$) & $80/200$ & $40.0$ & $33.5\text{--}46.9$ \\
\bottomrule
\end{tabular}
\end{table}

Success increased monotonically with the number of executed actions, from $16.0\%$ with $H_t = 2$ to $40.0\%$ with $H_t = 15$. The uncertainty-adaptive setting ($\beta = 40.4$) reached $30.5\%$, comparable to a fixed horizon of $8$ actions and below the near-full-chunk setting. Differences between adjacent rows lie within the statistical uncertainty, but the overall trend is consistent.

These results do not support a benefit of uncertainty-adaptive execution in this setting. The best result was obtained when Eq.~\ref{eq:dynamic_horizon} executed almost the entire chunk, which effectively disabled its adaptivity, whereas the configuration in which $H_t$ responded to $\bar{\sigma}_{t+1}$ performed worse. A likely contributing factor is that each new chunk is generated from fresh noise $\xi$, so that frequent replanning can switch between different, individually valid realizations of the same motion; executing a longer prefix reduces the number of such switches. We have not verified this explanation. It also remains open whether a better-calibrated uncertainty estimate, a different form of Eq.~\ref{eq:dynamic_horizon}, or a mechanism that keeps consecutive chunks consistent would change this outcome, and whether adaptive horizons become beneficial in the dynamic settings that motivated them (Section~\ref{sec:evaluation}).

\subsection{Latency Benchmarks}
\label{subsec:latency_results}

Section~\ref{subsec:latency_bound} reports per-stage latency measurements of the reference implementation on an Apple M4 (Table~\ref{tab:latency}); we summarize them here in the context of the LIBERO-10 experiments. For $K \in \{1, 2, 3\}$ and $H \le 64$, a complete decision step took $10.5\text{--}14.2$~ms at batch size one, of which $6.2\text{--}9.5$~ms was spent in the four networks on the dependency chain. At $K = 2$ and $H = 16$, the values used in the LIBERO-10 experiments, the complete step took $12.1$~ms. The sequential depth was $K + 3$ in all measured configurations, the encoder, selector, and predictor times did not change appreciably with $H$, and the time of one solver evaluation remained within $10\%$ of its value at $H = 8$ up to $H = 64$, consistent with Proposition~\ref{prop:complexity}.

Two points limit how these figures relate to the experiments above. First, the latency benchmark used the default module sizes with a single $224 \times 224$ view, whereas the LIBERO-10 models used $128 \times 128$ frames, two views in the two-camera case, and the settings of Table~\ref{tab:libero_setup}. We have not separately benchmarked the released checkpoints, although the released benchmark script accepts the configuration of a trained model. Second, in the LIBERO-10 experiments a decision step was followed by the open-loop execution of $H_t = 15$ actions, so the networks were evaluated once per $15$ control steps, and the simulated episodes did not run in real time. The measured latency therefore indicates the computational headroom of the pipeline rather than a constraint that was active in these experiments. Measurements on robotic deployment hardware and under real-time closed-loop operation remain future work.

\subsection{Future Work}
\label{sec:evaluation}

\begin{figure*}[htbp]
\centering
\resizebox{0.95\textwidth}{!}{%
\begin{tikzpicture}[
    >=stealth,
    box/.style={draw, rounded corners, thick, minimum width=3.5cm, minimum height=1cm, align=center, fill=blue!5, font=\small},
    data/.style={draw, minimum width=2.8cm, minimum height=0.8cm, align=center, fill=gray!10, font=\small},
    env/.style={draw, rounded corners, thick, minimum width=8cm, minimum height=1.2cm, align=center, fill=green!5, font=\small\bfseries},
    metric/.style={draw, dashed, thick, rounded corners, minimum width=2.5cm, minimum height=1cm, align=center, fill=red!5, font=\small},
    concat/.style={draw, rounded corners, thick, minimum width=9cm, minimum height=1cm, align=center, fill=yellow!10, font=\small}
]

\node[metric] (perturb) at (4, 1.8) {Dynamic Perturbations\\(Lighting Shifts, Moving Targets)};

\node[env] (sim) at (4, 0) {GPU-Accelerated Simulators\\(Isaac Sim / ManiSkill3)};
\draw[->, thick, dashed] (perturb.south) -- (sim.north);

\node[data] (img) at (0.5, -2.2) {Visual Stream ($I_t$, $I_g^{(m)}$)\\@ 60--100 Hz};
\node[data] (prop) at (7.5, -2.2) {Proprioception ($S_t$)\\@ 100--500 Hz};

\draw[->, thick, rounded corners] (sim.south) -- (4, -1.1) -| (img.north);
\draw[->, thick, rounded corners] (sim.south) -- (4, -1.1) -| (prop.north);

\node[concat] (agent) at (4, -4.5) {VPA Inference Pipeline\\$K{+}3$ Sequential Evaluations per Decision Step};

\draw[->, thick] (img.south) -- (img.south |- agent.north);
\draw[->, thick] (prop.south) -- (prop.south |- agent.north);

\node[data] (chunk) at (4, -6.5) {Action Chunk $A_t \in \mathbb{R}^{H \times d_a}$\\via CFM Solver ($K$ steps)};
\node[box] (executor) at (4, -8.3) {Low-Level Motor Controller\\Open-Loop Execution of First $H_t$ Steps};

\draw[->, thick] (agent.south) -- (chunk.north) node[midway, right, font=\footnotesize] {Integrate $v_\theta$};
\draw[->, thick] (chunk.south) -- (executor.north);

\draw[->, thick, rounded corners] (executor.west) -- (-4.2, -8.3) -- (-4.2, 0) -- (sim.west);
\node[font=\footnotesize, align=right, left=0.1cm] at (-4.2, -4.15) {Continuous\\Motor Commands\\($a_t \in A_t$)};

\node[metric, align=left] (eval) at (11.8, -4.5) {\textbf{Evaluation Logging}\\$\bullet$ Task Success (\%)\\$\bullet$ Latency Breakdown (ms)\\$\bullet$ Control Frequency (Hz)\\$\bullet$ $v(Z)$, $c(Z)$, Effective Rank\\$\bullet$ Phase Residuals $r_k$\\$\bullet$ Encoder Constants $\hat{\mu}_E, \hat{L}_E, \hat{\delta}_\eta$};
\draw[->, thick, dashed] (agent.east) -- (eval.west);

\draw[->, thick, dashed, rounded corners] (sim.east) -| node[pos=0.25, above, font=\footnotesize] {Privileged State $y_t$} (eval.north);

\end{tikzpicture}
} 
\caption{Planned Closed-Loop Evaluation in Simulation (Future Work). To benchmark high-frequency reactive control, the Goal-Conditioned VPA framework will be deployed in GPU-accelerated simulators subjected to dynamic physical and visual perturbations. High-frequency sensory streams are processed by the VPA pipeline, whose sequential depth of $K + 3$ network evaluations per decision step is independent of the action horizon (Proposition~\ref{prop:complexity}), and the resulting action chunks are executed open-loop for an uncertainty-dependent number of steps $H_t$. This setting is intended to test, under realistic closed-loop conditions, whether VPA avoids the low inference rates reported for autoregressive VLA architectures, such as about $1\text{--}5$~Hz for RT-2~\cite{brohan2023rt}. Because the simulator exposes the privileged physical state $y_t$, the evaluation will also log the quantities required to test the conditional analysis of Section~\ref{subsec:control_stability}.}
\label{fig:evaluation_implementation}
\end{figure*}
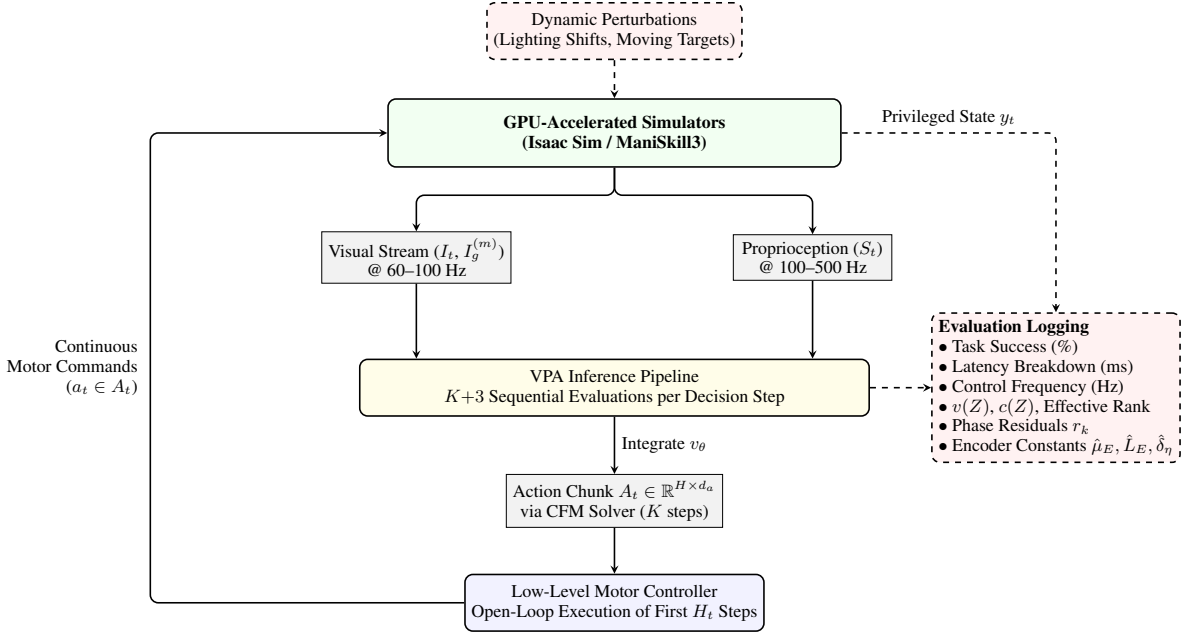

The experiments above establish that the pipeline can be trained and executed in a standard benchmark setting, but they leave open several questions that bear directly on the claims of this paper. First, the LIBERO-10 results should be strengthened by following the standard protocol, with several training seeds, $50$ trials per task, all LIBERO suites, and baselines evaluated under identical conditions. Initializing $E_\psi$ from a pre-trained self-supervised encoder \cite{oquab2023dinov2}, which the implementation already supports, and replacing the gripper-derived primitive labels with annotated or learned segmentations are natural next steps. Second, multi-milestone tasks ($M > 1$) are required to exercise the sequencing mechanism of Section~\ref{subsec:primitive_selector}, including the per-milestone time-out deferred there. Third, the properties that motivated the design, namely reactive control under dynamic perturbations and the conditional analysis of Section~\ref{sec:theoretical_analysis}, require environments that support such perturbations and expose the privileged physical state $y_t$.

Section~\ref{sec:theoretical_analysis} establishes two kinds of properties: properties that follow from the architecture and its objectives (Propositions~\ref{prop:complexity} and~\ref{prop:anti_collapse}), and properties that hold only under regularity assumptions on the learned encoder (Proposition~\ref{prop:phase_bound}). The latency measurements of Section~\ref{subsec:latency_bound} provide initial support for the former on one hardware platform; the latter remain untested. We therefore plan an evaluation that benchmarks high-frequency reactive control and visual robustness, areas in which current autoregressive Vision-Language-Action (VLA) models often experience difficulties, together with targeted measurements of the theoretical quantities introduced in Section~\ref{sec:theoretical_analysis}.

\subsubsection{Simulation Environments}
\label{subsec:eval_environments}
To benchmark the system reproducibly prior to physical hardware deployment, we plan to use established GPU-accelerated simulation platforms:
\begin{itemize}
    \item \textbf{NVIDIA Isaac Sim} \cite{nvidia_isaacsim}: This platform is selected for its PhysX physics engine and ray-traced rendering. It is suited to testing robustness against visual distribution shifts, such as dynamic lighting or reflections, and to assessing whether the regularized JEPA representations remain informative without pixel reconstruction.
    \item \textbf{ManiSkill3} \cite{tao2024maniskill3}: This simulator is selected for its standardized manipulation tasks and parallelized tensor-based API. It allows for the rapid evaluation of continuous action chunking and neuro-symbolic transitions across numerous episodic rollouts.
\end{itemize}
Both platforms expose the ground-truth physical state $y_t$ (robot configuration and object poses), which is unavailable on physical hardware and is required for the measurements of Section~\ref{subsec:eval_theory}.

\subsubsection{Planned Benchmark Tasks}
\label{subsec:eval_tasks}
The planned evaluation focuses on dynamic, non-stationary environments in which the low per-step latency measured in Section~\ref{subsec:latency_bound} may translate into a measurable advantage over autoregressive VLAs, for which inference rates of about $1\text{--}5$~Hz have been reported for RT-2 \cite{brohan2023rt}.

\textbf{1. High-Speed Reactive Pick-and-Place:}
Unlike static tabletop manipulation, this task involves moving targets, such as items on a moving conveyor belt or objects physically perturbed mid-grasp. Autoregressive models often experience action lag in these scenarios, executing stale trajectories based on outdated observations. The Flow-Matching solver ($\pi_\theta^l$) will be evaluated on its ability to update action chunks at $>50$~Hz to track and grasp moving targets. In light of Section~\ref{subsec:horizon_ablation}, this task will also test whether frequent replanning can be made beneficial when the target moves.

\textbf{2. Reactive Obstacle Avoidance:}
In this task, the robot is instructed to reach a visual goal, but dynamic obstacles are introduced into the workspace mid-trajectory. This evaluates the neuro-symbolic primitive selector ($\pi_\phi^h$) and the discrete sub-goal thresholding ($\Vert z_t - z_g^{(m_t)} \Vert_2 < \tau$). If an obstacle blocks a trajectory, the system is expected to halt milestone progression or switch primitives to mitigate error accumulation; the frequency and duration of such halts also serve to characterize pointer stalls (Section~\ref{subsec:primitive_selector}).

\textbf{3. Visual Nuisance Generalization:}
To test the nuisance robustness motivated in Sections~\ref{subsec:visual_encoding} and~\ref{subsec:anti_collapse}, and to estimate the nuisance slack $\delta_\eta$ of Section~\ref{subsec:control_stability}, the agent will be trained in a visually sterile environment and tested in environments with visual distractors, such as changing background textures, moving camera perspectives, and variable illumination. The metric of interest is the zero-shot success rate decay compared to pixel-reconstruction baselines.

\textbf{4. Language Intent Disambiguation:}
To evaluate the system's semantic comprehension, this task presents environments in which visual goals are ambiguous without language context (e.g., distinguishing between an instruction to move a cup and an instruction to throw it away). This benchmark tests whether the text embedding $\ctext$ produced by the lightweight frozen encoder provides sufficient conditioning for the primitive selector and the flow-matching solver, which are the modules that receive it (Eqs.~\ref{eq:primitive_selection} and~\ref{eq:conditioning_vector}), to resolve the overarching task intent without a large language model.

\subsubsection{Verification of Theoretical Properties and Ablations}
\label{subsec:eval_theory}
Beyond task-level performance, we plan the following measurements, each of which targets a specific statement or design choice of Sections~\ref{sec:architecture} and~\ref{sec:theoretical_analysis}:
\begin{enumerate}
    \item \textbf{Latency on deployment hardware (Proposition~\ref{prop:complexity}):} Extend the measurements of Section~\ref{subsec:latency_bound} to robotic deployment hardware and to trained configurations, including the two-camera encoder, by measuring $T_{\text{VPA}}$ and its components $t_E$, $t_h$, $t_P$, and $t_l(H)$ (Eq.~\ref{eq:vpa_latency}) as $H$ and $K$ are varied, together with the end-to-end closed-loop control frequency.
    \item \textbf{Encoder regularity (Assumptions~\ref{assump:upper} and~\ref{assump:lower}):} Using the privileged state $y_t$, sample state pairs under controlled nuisance variation and estimate $\hat{\mu}_E$, $\hat{L}_E$, and $\hat{\delta}_\eta$ from the lower and upper envelopes of $\Vert z - z' \Vert_2$ as a function of $\Vert y - y' \Vert_2$ (e.g., via quantile regression), with and without spectral normalization of the encoder.
    \item \textbf{Phase-onset residuals (Proposition~\ref{prop:phase_bound}):} On tasks with $M > 1$, record the residual $r_k$ at every milestone transition and compare its empirical distribution with the conditional bound $(\tau^\star + \hat{\delta}_\eta)/\hat{\mu}_E$, together with the rates of premature transitions, missed transitions, and stalls under a per-milestone time-out.
    \item \textbf{Threshold calibration (Eq.~\ref{eq:transition_rule}):} Compare a fixed absolute threshold with the variance-margin-scaled threshold of Eq.~\ref{eq:dynamic_threshold} across training checkpoints, and sweep the tolerance $\kappa$ to trace the trade-off between soundness and completeness described in Section~\ref{subsec:control_stability}.
    \item \textbf{Horizon stability (Eq.~\ref{eq:dynamic_horizon}):} Building on Section~\ref{subsec:horizon_ablation}, compare the raw uncertainty score $\sigma_{t+1}$ with the filtered signal $\bar{\sigma}_{t+1}$ in terms of the frequency of horizon switches and the magnitude of action discontinuities at chunk boundaries, and examine whether keeping consecutive chunks consistent changes the effect of the executed horizon on task success.
\end{enumerate}

\subsubsection{Baselines and Evaluation Metrics}
\label{subsec:eval_metrics}
The VPA architecture will be benchmarked against autoregressive VLA models (e.g., RT-2 \cite{brohan2023rt}), generalist policies with continuous action heads (e.g., Octo \cite{octo2024} and $\pi_0$ \cite{black2024pi}), and diffusion-based behavior cloning policies \cite{ho2020denoising, chi2025diffusion}, under identical observations, training data, and evaluation protocols. Including continuous-head baselines isolates the contribution of the lightweight latent pipeline from that of non-autoregressive action generation, since the latter reduction in sequential depth is shared by these methods (Section~\ref{subsec:latency_bound}). The primary evaluation metrics will include:
\begin{enumerate}
    \item \textbf{Task Success Rate (\%):} Measured across static and dynamic environments, with several training seeds.
    \item \textbf{End-to-End Inference Latency (ms):} Wall-clock time from image capture to motor command generation, reported together with its per-module breakdown and the deployment hardware.
    \item \textbf{Control Frequency (Hz):} The maximum stable closed-loop execution rate.
    \item \textbf{Representation Health:} The variance and covariance penalties $v(Z)$ and $c(Z)$, together with the effective rank of $C(Z)$ \cite{roy2007effective}, evaluated for the frozen visual encoder under visual shift.
    \item \textbf{Phase-Transition Fidelity:} The distribution of phase-onset residuals $r_k$ and the rates of premature transitions, missed transitions, and stalls.
    \item \textbf{Execution Smoothness:} The frequency of changes in the executed horizon $H_t$ and the magnitude of action discontinuities at chunk boundaries.
\end{enumerate}

\section{Conclusion}
\label{sec:conclusion}

The pursuit of generalized embodied intelligence has recently been supported by the scaling of autoregressive Vision-Language-Action (VLA) models. However, the sequential token generation on which many of these models rely makes them challenging to deploy for the high-frequency, reactive demands of physical robotic control.

In this paper, we introduced the Goal-Conditioned Vision-Predictor-Action (VPA) framework, proposing a shift from text-conditioned autoregressive decoding to latent-predictive continuous control, and implemented it as an open-source reference implementation. By decoupling high-level task logic from low-level motor execution and utilizing a lightweight frozen text encoder to capture user instructions, the VPA framework aims to retain semantic task comprehension without relying on computationally heavy Large Language Models (LLMs).

Our analysis establishes three results at different levels of strength. First, formulating multi-step action chunking through Conditional Flow-Matching (CFM) reduces the number of sequential network evaluations per decision step from $\Theta(H \cdot d_a)$ for autoregressive token decoding to $K + 3$, a quantity independent of the action horizon $H$; the arithmetic cost of each solver evaluation nevertheless still scales with the chunk dimensionality $H \cdot d_a$, so the per-cycle cost is not $\mathcal{O}(1)$ in $H$. Second, the VICReg variance and covariance penalties exclude fully collapsed representations from the global minimizers of the joint-embedding objective whenever a non-degenerate representation attains a lower objective value, and they exclude dimensional collapse on batches where both penalties vanish; the realized variance margin further provides a principled scale for calibrating the milestone completion threshold $\tau$. Third, discrete sub-goal thresholding bounds the physical residual at each phase transition by $(\tau + \delta_\eta)/\mu_E$, but only under a co-Lipschitz condition on the encoder with respect to task-relevant physical state. Because Vision Transformers are not globally Lipschitz without explicit constraints, and because this lower bound cannot be enforced through architectural constraints alone, we regard the containment of error across task phases as an empirical expectation rather than a formal guarantee.

\paragraph{Initial empirical validation.} The reference implementation provides a first proof of concept for the framework (Section~\ref{sec:experiments}). The trainable modules contain fewer than $20$M parameters, and the frozen CLIP text encoder ($63.4$M parameters) is evaluated only once per episode, outside the per-step dependency chain. Each model trained in under three hours on a single NVIDIA L4 GPU. On an Apple M4, a complete decision step of the single-camera reference configuration took $10.5\text{--}14.2$~ms for $K \in \{1, 2, 3\}$ and $H \le 64$, and the per-stage timings were consistent with the horizon-independent sequential depth of Proposition~\ref{prop:complexity} (Section~\ref{subsec:latency_bound}). On LIBERO-10, with a vision encoder trained from random initialization on the demonstrations alone, the policy succeeded in $40.0\%$ of $200$ evaluation episodes ($95\%$ interval $33.5\text{--}46.9\%$) when observing a third-person and a wrist-mounted camera, compared with $10.5\%$ with the third-person camera alone. In an ablation of the executed horizon, executing nearly the full action chunk before replanning gave the best results. We regard these results as promising but preliminary: they indicate that a latent control loop without a language-model decoder can be trained from demonstrations to complete long-horizon manipulation tasks in a meaningful fraction of trials, while its per-step computation remains small enough to run at low latency on commodity hardware.

\paragraph{Limitations.} The repertoire of high-level behaviors is bounded by the fixed primitive dictionary $\mathcal{U}$, and zero-shot synthesis of skills outside $\mathcal{U}$ is not addressed by the present formulation. The empirical results come from a single training run per configuration with $20$ trials per task. They were not compared with baselines under the same protocol, and they remain below success rates previously reported on this suite. The experiments used a single goal image per task ($M = 1$) and primitive labels derived heuristically from the gripper command, so the multi-milestone sequencing mechanism and the conditional analysis of Proposition~\ref{prop:phase_bound} were not exercised; the encoder constants $\mu_E$, $L_E$, and $\delta_\eta$ and the phase-onset residuals $r_k$ have not yet been measured. The uncertainty-adaptive execution horizon of Eq.~\ref{eq:dynamic_horizon} did not improve task success in our ablation (Section~\ref{subsec:horizon_ablation}). Latency was measured for a single-camera reference configuration on one hardware platform, not for the released checkpoints or on robotic deployment hardware, and closed-loop control frequency under real-time operation was not measured. Finally, the single-frame observation of the original formulation proved insufficient in our experiments, and the reported success rate depends on extending $I_t$ to two camera views.

We therefore view VPA as an analytically characterized architecture with an initial working implementation, whose theoretical properties are stated together with the assumptions they require and whose empirical support is, so far, limited to the setting reported here. The next steps are to evaluate the framework under the standard LIBERO protocol with baselines, on multi-milestone tasks, and in the dynamic simulation environments outlined in Section~\ref{sec:evaluation}. We expect the explicit separation between unconditional and conditional properties to provide a sound basis for this work, and we hope that the released code and checkpoints (Section~\ref{sec:availability}) will help others reproduce, test, and improve upon these results.

\bibliographystyle{unsrt}
\bibliography{references}

\end{document}